\pdfoutput=1
\documentclass[sigplan,10pt]{acmart}
\usepackage{algorithm}                 %
\usepackage[noend]{algpseudocode}      %
\usepackage{balance}                   %
\usepackage{xspace}
\usepackage{comment}
\usepackage{enumitem}

\usepackage{pgfplots}
\usepackage{fontawesome}
\usepackage{wrapfig}

\pgfplotsset{compat=1.18}
\pgfplotsset{every axis/.append style={axis line style={line width=0.5pt}, every tick/.append style={line width=0.5pt}}}
\pgfplotsset{every axis/.append style={error bars/error bar style={line width=0.5pt, black}, error bars/error mark options={rotate=90, line width=0.5pt, black}}}
\pgfplotsset{every axis legend/.append style={line width=0.5pt}}
\usepgfplotslibrary{groupplots}
\usepgfplotslibrary{statistics}

\usepackage{subcaption}

\usepackage[capitalize]{cleveref}
\usepackage{booktabs}
\usepackage{array}

\usepackage{listings}
\usepackage{xcolor}

\usepackage[frozencache,cachedir=minted-cache]{minted}
\usepackage{tikz}
\usetikzlibrary{
    arrows.meta,      %
    positioning,      %
    shapes.geometric, %
    shadows,          %
    shadows.blur,
    decorations.pathmorphing, %
    calc,             %
    fit,              %
    backgrounds,       %
    tikzmark,
    positioning,
    chains
}

\usepackage{tcolorbox}
\tcbuselibrary{skins} %
\tcbset{
    mygraybox/.style={
        colback=gray!15, %
        colframe=gray!60, %
        rounded corners, %
        boxsep=5pt, %
        arc=5pt, %
        top=3pt, bottom=3pt, left=3pt, right=3pt, %
        boxrule=0.5pt %
    }
}

\definecolor{codegreen}{rgb}{0,0.6,0}
\definecolor{codegray}{rgb}{0.5,0.5,0.5}
\definecolor{codepurple}{rgb}{0.58,0,0.82}
\definecolor{backcolour}{rgb}{0.95,0.95,0.92}

\lstdefinestyle{mystyle}{
    backgroundcolor=\color{backcolour},
    commentstyle=\color{codegreen},
    keywordstyle=\color{magenta},
    numberstyle=\tiny\color{codegray},
    stringstyle=\color{codepurple},
    basicstyle=\ttfamily\footnotesize,
    breakatwhitespace=false,
    breaklines=true,
    captionpos=b,
    keepspaces=true,
    numbers=left,
    numbersep=5pt,
    showspaces=false,
    showstringspaces=false,
    showtabs=false,
    tabsize=2
}
\newcommand{\code}[1]{\texttt{\detokenize{#1}}}

\newcommand{\tsan}{\textsc{ThreadSanitizer}\xspace}
\newcommand{\fasttrack}{\textsc{FastTrack}\xspace}

\newtheorem{remark}{Remark}
\newenvironment{proofsketch}{%
  \renewcommand{\proofname}{Proof Sketch}%
  \proof
}{%
  \endproof
}
\newcommand{\Locks}{{\mathcal{L}}}
\newcommand{\Globals}{{\mathcal{G}}}
\newcommand{\Locals}{{\mathcal{V}}}
\newcommand{\Formals}{{\mathcal{R}eg}^{\textsf{Form}}}
\newcommand{\Registers}{{\mathcal{R}eg}}
\newcommand{\memAccessInstructions}{{\mathcal{M}}}
\newcommand{\memInstr}{{\mathcal{I}}}
\newcommand{\memInstrSuff}{{\underline{\mathcal{I}}}}
\newcommand{\memInstrRed}{\mathcal{R}}
\newcommand{\sem}[1]{\llbracket #1 \rrbracket}
\newcommand{\Functions}{{\mathcal{F}}}
\newcommand{\BBs}{{\mathbf{BB}}}
\newcommand{\func}{\textsf{func}}
\newcommand{\bblock}{\textsf{block}}
\newcommand{\mainFun}{{\textsf{main}}}
\newcommand{\lock}[1]{\texttt{lock(#1)}}
\newcommand{\unlock}[1]{\texttt{unlock(#1)}}
\newcommand{\create}[1]{\texttt{create(#1)}}
\newcommand{\writes}[1]{\textsf{writes}(#1)}
\newcommand{\locs}[1]{\textsf{locs}(#1)}
\newcommand{\Instructions}{{\mathcal{I}\!nstr}}

\newcommand{\RQ}[1]{\textbf{RQ#1}}

\title[Instrumentation Optimization for Practical Dynamic Race Detection]{Instrumentation Optimization for\texorpdfstring{\\}{ }Practical Dynamic Race Detection}
\author{Alexey Paznikov}
\affiliation{%
  \institution{National University of Singapore}
  \city{Singapore}
  \country{Singapore}
}
\email{paznikov@nus.edu.sg}

\author{Michael Schwarz}
\affiliation{%
  \institution{National University of Singapore}
  \city{Singapore}
  \country{Singapore}
}
\email{m.schwarz@nus.edu.sg}

\author{Andrey Kogutenko}
\affiliation{%
  \institution{Saint Petersburg Electrotechnical University "LETI"}
  \city{Saint Petersburg}
  \country{Russia}
}
\email{andrey.a.kogutenko@gmail.com}

\author{Yaroslav Osipov}
\affiliation{%
  \institution{Saint Petersburg Electrotechnical University "LETI"}
  \city{Saint Petersburg}
  \country{Russia}
}
\email{yaroslav.o@keemail.me}

\author{Umang Mathur}
\affiliation{%
  \institution{National University of Singapore}
  \city{Singapore}
  \country{Singapore}
}
\email{umathur@nus.edu.sg}

\renewcommand{\shortauthors}{Paznikov et al.}

\begin{abstract}
Dynamic data race detectors are successful at flagging concurrency bugs,
but their high runtime overhead limits adoption. 
This overhead stems primarily from pervasive instrumentation of memory accesses -- a significant fraction of which is for safe accesses. 
We address this inefficiency through a compiler-integrated approach that identifies safe accesses using static analysis and eliminates instrumentation for these.
We strike a careful balance between analysis precision and cost, to ensure applicability to real-world applications. We show that our optimizations never cause \tsan to miss races it would otherwise detect.
We further observe that many checks are \emph{redundant}: They only flag
a race if a preceding check has already flagged an
\emph{equivalent} race for the same memory location.
We %
show that, when reporting one representative per class is sufficient, even more instrumentation can be eliminated.
Lastly, we introduce a dynamic technique that skips instrumentation when there is
only one thread at runtime, even when this could not be determined statically.
We integrate these static analyses and program transformations
into the \tsan race detector based on \textsc{LLVM}. 
Our experiments on a diverse suite of real-world applications, 
including \textsc{Memcached}, \textsc{Redis}, \textsc{SQLite}, \textsc{MySQL}, and \textsc{Chromium}, demonstrate that our optimizations reduce the number of instrumentation calls at runtime by up to 91\% and yield application-level speedups of up to $1.91\times$, a geometric-mean Chromium speedup of $1.25\times$ on Blink and $1.14\times$ on Speedometer~3.0, and peak speedups reaching $2.28\times$ on Memcached with two server threads.
Our optimizations are in the process of being upstreamed into \tsan.
\end{abstract}

\copyrightyear{2026}
\acmYear{2026}
\setcopyright{cc}
\setcctype{by}
\acmConference[SIGOPS ATC 2026]{The 2026 ACM SIGOPS Annual Technical Conference}{November 15--18, 2026}{Hong Kong, China}
\acmBooktitle{The 2026 ACM SIGOPS Annual Technical Conference (SIGOPS ATC 2026), November 15--18, 2026, Hong Kong, China}
\acmDOI{10.1145/3828523.3857050}
\acmISBN{979-8-4007-2818-1/2026/11}



\makeatletter
\patchcmd{\@mkbibcitation}{\ref{TotPages}~\@pages@word}{\getpagerefnumber{paper-end}~\@pages@word}{}{\errmessage{bibstrip patch 1 failed}}
\patchcmd{\@mkbibcitation}{\ref{TotPages}~\@pages@word}{\getpagerefnumber{paper-end}~\@pages@word}{}{\errmessage{bibstrip patch 2 failed}}
\makeatother
\begin{document}
\maketitle
\section{Introduction}
Programmers use concurrency to harness the %
power of multi-core architectures. %
However, developing correct and efficient concurrent programs is 
challenging, with data races standing out as an insidious class of bugs. 
A data race occurs when two %
threads concurrently access the same memory location, at least one of these accesses is a write, and the accesses are not ordered by synchronization operations. 
Races can have wide-ranging consequences, 
including crashes, silently corrupting data, %
and security vulnerabilities. 
Because races and their consequences manifest nondeterministically,
they are exceptionally difficult to reproduce and debug.

Dynamic race detectors such as \tsan~\cite{serebryany2009threadsanitizer} and 
\textsc{Helgrind}~\cite{helgrindManual}
have become the de-facto standard for runtime race detection.
They instrument memory accesses and synchronization operations to track
causality via the \emph{happens-before}~\cite{lamport2019time} partial order,
flagging unordered conflicting accesses.
Such detectors are sound (no false positives) and effective in practice.
However, they introduce significant runtime overheads,
with the official \tsan documentation describing $5\times$--$15\times$ slowdowns as typical~\cite{TSanDoc}.
On \textsc{Chromium}, for instance, \tsan reduces the \textsc{Speedometer~3.0} score from 2.39 to 0.25 --- a $9.6\times$ slowdown.
Such overheads hinder adoption in continuous integration, interactive debugging, and large applications.

Proposals to reduce the overhead of dynamic race detectors range from
algorithmic improvements to underlying data structures
such as the use of vector clocks (\textsc{Djit+}),
epoch optimization (\fasttrack), to
optimal timestamping data structures~\cite{flanagan2009fasttrack,pozniansky2007multirace, treeClocksMathur2022}
and cheaper synchronization for metadata updates~\cite{WoodCBG17}.
Other approaches include
sampling~\cite{bond2010pacer,marino2009literace,RPT2023,erickson2010effective,UTrack2025,racemob2013} 
to check fewer accesses, inherently trading completeness for performance.

There is also some work on leveraging static analysis to reduce the overhead of dynamic race detection~\cite{flanagan2013redcard,di2016accelerating,choi2002efficient,vonPraun2003static,rhodes2017bigfoot}, with work featuring
expensive whole-program analyses,
focusing only on narrow sets of specific patterns (often targeting managed languages such as Java),
or not easily integrated into development workflows.

We take up the idea of \textbf{combining static and dynamic analysis} and propose a comprehensive framework built around
the commercial-grade dynamic race detector \tsan which ships with \textsc{LLVM}, and
which carefully balances the precision of the static analyses with their cost, for integration into real-world development workflows.
We identify which memory accesses are safe to leave uninstrumented,
with a guarantee that no races are missed.
We instantiate our framework with five complementary static analyses
to achieve both instruction- and location-level pruning across the entire program,
while preserving completeness. We also demonstrate how a dynamic optimization to skip instrumentation
for accesses in single-threaded contexts can further reduce the overhead of \tsan for some workloads at a per-access cost of one atomic read and branch.

The ideal analysis would identify exactly the accesses that can race; this is undecidable.
We instead compute over-approximations of the set of potentially racy accesses.
There is a balance to strike here: the more precise an analysis is, the more
instrumentation it will be able to eliminate, but the more costly it will be
to run this analysis at compile time.
We opt for cheap analyses that can be run as part of the compilation
pipeline, and base our choice of analyses on empirical observations:
For a significant fraction of memory access instructions, there are
straightforward reasons why they cannot participate in any data race.
This can, e.g., be because the accessed memory locations are all thread-local,
are consistently protected using locks, or because the access always
occurs before any additional threads are created.
The following four analyses compute such over-approximations:
\begin{description}[leftmargin=0.5cm]

    \item[\rm\textbf{Single-threaded context (STC)}.]
    Analysis to identify functions only executing in single-threaded mode. %
    Accesses in such functions need not be instrumented.

    \item[\rm\textbf{Escape analysis (EA)}.]
    \emergencystretch=0pt %
    Field-sensitive, interprocedural analysis that identifies objects confined to a single thread (i.e., do not escape) and thus cannot participate in a data race.

    \item[\rm\textbf{Lock-ownership analysis (LO)}.]
    \emergencystretch=2em
    Interprocedural dataflow analysis that identifies global variables consistently protected by specific locks.
    
    \item[\rm\textbf{Single-writer/multiple-reader analysis (SWMR)}.]
    \hspace{0pt}Identifies global variables written only in single-threaded contexts. During concurrent execution, these variables are read-only and need no instrumentation.
\end{description}
Eliminating such instrumentation does not affect the detector's output beyond reducing its overhead.

Many remaining checks will flag a race if and only if another, earlier check on the same location does.
In particular, when an access $e$ to a memory location
is preceded or succeeded by another access $e'$ to the same memory location
without any intervening synchronization,
it suffices to instrument only one if the user is satisfied 
with one representative race flagged in such cases;
our fifth optimization eliminates such \emph{redundant} instrumentation.
\begin{description}[leftmargin=0.5cm]
    \item[\rm \textbf{Dominance-based redundancy elimination (DE)}.] 
    To\\ identify redundant instrumentation for memory accesses that are \emph{(post)dominated} by instrumented accesses to the same location, provided all paths between them are free of synchronization and external calls.
\end{description}
Our motivating example (\cref{fig:unified_example}) shows a common single-threaded initialization phase followed by multi-threaded processing.
Bold comments identify accesses that need not be instrumented and the responsible analysis.

\begin{figure*}
\centering
\begin{minipage}[t]{0.32\textwidth}
\vspace{0pt}
\begin{minted}[linenos,breaklines=true,xleftmargin=18pt,numbersep=4pt,fontsize=\tiny,escapeinside=||]{C}
pthread_mutex_t lock;
int g_counter = 0;    // Protected global
int g_config;         // SWMR global
int g_sum;
void* g_queue[1024];  // Global work queue

typedef struct {
  int id;             // Thread-local field
  int payload;        // Field will escape
} req_t;

// Called by main before creating threads
void initialize_system() {
\end{minted}
\end{minipage}\hfill
\begin{minipage}[t]{0.32\textwidth}
\vspace{0pt}
\begin{minted}[linenos,breaklines=true,xleftmargin=18pt,numbersep=4pt,fontsize=\tiny,firstnumber=14,escapeinside=||]{C}
  g_config = 42;       // (1) |\textbf{Safe by STC}|
}

// Helper function for thread workers
void process_req_internals(req_t* req) {
  req->id++;   // (2a) |\textbf{Safe by EA}|
}

// Main function for each worker thread
void* worker_thread(void* arg) {
  pthread_mutex_lock(&lock);
  g_counter++;         // (4) |\textbf{Safe by LO}|
  pthread_mutex_unlock(&lock);
\end{minted}
\end{minipage}\hfill
\begin{minipage}[t]{0.32\textwidth}
\vspace{0pt}
\begin{minted}[linenos,breaklines=true,xleftmargin=18pt,numbersep=4pt,fontsize=\tiny,firstnumber=27,escapeinside=||]{C}
  req_t req;
  req.id = g_config;   // (3) |\textbf{Safe by SWMR}|
  req.payload = create_payload();
  process_req_internals(&req);
  g_queue[...] = &req.payload; // (2b) 
                      // Field escapes

  g_sum = 0;
  for (int i = 0; i < SIZE; ++i)
    // (5) dominated by 'g_sum = 0' |\textbf{(DE)}| 
    g_sum += get_value(g_queue[i]); 
}
\end{minted}
\end{minipage}

\caption{Unified example illustrating all five optimizations.}
\Description{C program in three columns that marks one access for each analysis: the write of g\_config before any thread starts (STC), the increment of a request's id field, which stays thread-local although its payload field escapes (EA), the increment of g\_counter under a lock (LO), the read of g\_config in a worker thread (SWMR), and the updates of g\_sum in a loop, dominated by the write g\_sum = 0 (DE).}
\label{fig:unified_example}
\end{figure*}
Based on the observation that the \textbf{STC} analysis can be quite imprecise for common patterns such as thread pools,
we propose a further optimization to reduce the overhead of instrumentation by performing a \emph{program transformation}.
\begin{description}[leftmargin=0.5cm]
    \item[\rm\textbf{Overhead reduction by program transformation.}]
    To\\ complement the static approaches, we also propose a lightweight program transformation
    to ensure instrumentation calls only happen when more than one thread is active, by maintaining a
    counter of the number of active threads and branching on it at access instructions. 
\end{description}
While this transformation introduces some overhead from branching and counter updates, experiments show it is small compared to instrumentation costs.

We implemented a comprehensive framework within \textsc{LLVM} for \tsan,
and instantiated it with the five static analyses described and the dynamic optimization.
On real-world applications, our approach achieves application-level speedups up to $1.91\times$,
a geometric-mean Chromium speedup of $1.25\times$ on Blink and $1.14\times$ on Speedometer~3.0, and peaks of $2.28\times$ on Memcached with two server threads.
Our optimizations are in the process of being upstreamed into \tsan.

\Cref{sec:background} gives background on dynamic race detection.
\Cref{sec:sufficient} characterizes sufficient instrumentation; \cref{sec:analyses} presents our static analyses.
\Cref{sec:beyond,sec:trans} cover redundancy elimination and the dynamic transformations.
\Cref{sec:pipeline_implementation} details the implementation;
\cref{sec:eval} presents the evaluation; \cref{sec:related,sec:conclusion} close the paper. The appendices are available as supplementary material.

\section{Dynamic Data Race Detection}\label{sec:background}
We provide an overview of \tsan, focusing on why the introduced instrumentation incurs significant performance overheads.  We then present the core fragment of an imperative language used to formalize our approach and supply the necessary definitions.

\subsection{\tsan and its Overhead}

The primary source of overhead in %
race detectors like \tsan~\cite{serebryany2009threadsanitizer, serebryany2011dynamic}
are the calls to runtime library (RTL) functions that are inserted before each memory access instruction.
A write \code{g = 5;} in the original program, for instance, is transformed into the sequence \code{__tsan_write(&g); g = 5;} in this so-called instrumentation step.\footnote{
    There are different RTL functions for different access sizes --- we ignore this detail here and in the rest of the paper.
} 
The same happens for synchronization primitives (e.g., mutex locks/unlocks).
At runtime, a happens-before~\cite{lamport2019time} algorithm based on vector clocks~\cite{pozniansky2007multirace} tracks the partial order of events across threads. This information is stored in shadow memory. Upon an instrumented access, the race detection logic checks the access's metadata against the state in shadow memory.
If a conflicting access is found that is not ordered by the happens-before relationship, a race is reported.
Thus, a simple memory access in the original program can translate to dozens of instructions at runtime, causing significant overheads.

\subsection{Concrete Semantics and Definitions}\label{ss:concrete-semantics}
\looseness=1
While our implementation targets \tsan and is integrated into
\textsc{LLVM}, for soundness proofs, we consider a set
of instructions sufficient to model relevant aspects of a small,
{C}-like language with threads, shared memory, and locks.
We assume the program is given as a set
of control flow graphs (CFGs), one per function.
Each function has local variables and virtual registers.
Local variables may have their address taken. They can escape their function or even their thread.
Virtual registers cannot have their address taken and always remain local.
Each function's CFG consists of basic blocks: sequences of instructions with no internal control flow except function calls.
Basic blocks are connected by edges labelled with guards, which are predicates over registers.
Instructions create and return from threads, (un)lock locks, access memory,
call functions directly or through pointers, and compute on registers.
The appendix summarizes our notation and instruction set in \cref{fig:definitions-notation,fig:instr-set}.
We assume programs to be type-correct and free of use-after-free on the heap or the stack;
like other compiler optimizations, ours assume the absence of this undefined behavior (\cref{app:uaf}).
We also assume that program execution starts from a function {\mainFun} which itself is not called in the program. %
The concrete semantics of a program is the set of all its traces, each given as a sequence of dynamic events.
A dynamic event consists of a thread identifier and the instruction executed by that thread.
For reads and writes through pointers, we explicitly record the accessed memory location in the dynamic event rather than reconstructing it from the trace.
We require traces to satisfy the usual well-formedness conditions: each lock is held by at most one thread, and each thread's steps follow a valid path in the CFG.
For variable values, we assume sequentially consistent execution, as is common in prior work.
We are interested in the detection of data races, more precisely in happens-before based detection.
This is based on weakening the total order $\leq$ on dynamic events to the happens-before partial order $\to_{hb}$ obtained 
as the reflexive and transitive closure of the union of:
\begin{itemize}[wide]
    \item program order $\to_{po}$: totally orders events of the same thread according to position in the trace
    \item lock order $\to_{lo}$: orders unlock events before subsequent lock events of the same lock
    \item thread creation order $\to_{cr}$: orders the thread creation event before first event of created thread
\end{itemize}
\looseness=0
\begin{definition}[Data Race]
Two dynamic events $e_1$ and $e_2$ form a \emph{data race}, when
they access the same memory, at least one of them is a write, and they are unordered w.r.t.\ $\to_{hb}$.
\end{definition}
We also call two memory access instructions $i_1$ and $i_2$ \emph{racy} in a trace $t$, if there are two dynamic events
corresponding to the respective instructions that form a data race.
\begin{remark}
    Our results remain sound when stronger versions of the happens-before order are used that, e.g., also account for synchronization
    arising out of synchronization primitives such as condition variables or semaphores.
    This justifies our choice to keep the set of considered instructions small.
\end{remark}
\section{Sufficient Instrumentation}
\label{sec:sufficient}
We lower dynamic race detector overhead by statically identifying accesses that cannot race
and instrumenting only the remainder.
The approach applies to any instrumentation-based detector satisfying the following assumptions.
\begin{definition}[Instrumented Memory Accesses]
    Let $\memAccessInstructions$ be the set of all memory access instructions. %
    We assume that the analyzer takes as an additional parameter a set %
    $\memInstr \subseteq \memAccessInstructions$, representing the instructions to be instrumented.
\end{definition}
Next, we state a completeness property, namely that the analyzer will
flag races involving \emph{instrumented} accesses upon observing an execution exhibiting them. %
\begin{definition}[Completeness w.r.t.\ Instrumented Accesses]\label{def:completeness}
    We call a data race detector \emph{complete w.r.t.\ a set of
    instrumented accesses $\memInstr$}, iff it flags all races involving only access instructions in $\memInstr$
    that occur in the executions it observes.
\end{definition}
\begin{remark}
    This is an idealized property.
    Production race detectors like \tsan are optimized for performance and deviate from it;
    they typically maintain only bounded shadow data per memory location, i.e., a fixed-size summary of recent accesses.
    Thus, races in a trace may be missed even when both accesses are instrumented.
    We nevertheless state our results in terms of \cref{def:completeness}, which cleanly captures the behavior of an idealized tool.
    Importantly, our arguments also apply under a weaker requirement satisfied by production tools: omitting instrumentation for accesses that cannot race must not cause the tool to miss races it otherwise reports.
\end{remark}
Setting $\memInstr = \memAccessInstructions$, i.e., instrumenting all memory accesses, yields a data race detector
that will flag all races in the executions it observes.
We seek a smaller instrumentation set without losing this property.
\begin{definition}[Sufficient Instrumentation Set]\label{def:sufficient-instr-set}
    For a program and its set of memory access instructions $\memAccessInstructions$, 
    the set of minimal required memory instrumentations $\memInstrSuff$ is the set of those access instructions 
    that are actually involved in data races in at least one of the executions of the program.
    We call any set $\memInstr$ with $\memInstrSuff \subseteq \memInstr$ a \emph{sufficient instrumentation set}.
\end{definition}
\begin{proposition}
    A dynamic data race detector that is complete w.r.t.\ instrumented accesses,
    when instantiated with some instrumentation set $\memInstr$ subsuming $\memInstrSuff$,
    enjoys the same completeness properties as one instrumenting all accesses. 
\end{proposition}
\begin{proof}
    Directly from \cref{def:completeness} and \cref{def:sufficient-instr-set}.
\end{proof}
Computing $\memInstrSuff$ precisely amounts to sound and complete
static race detection, which is undecidable.
We therefore compute overapproximations of \(\memInstrSuff\)
so fewer accesses need to be instrumented, reducing overhead without affecting the dynamic analyzer's completeness guarantees.

\section{Static Analyses to Overapproximate the Sufficient Instrumentation Set \texorpdfstring{$\memInstrSuff$}{}}\label{sec:analyses}
We combine four analyses to identify accesses that cannot race and need not be instrumented.
We focus on high-level ideas and defer formal treatment and proofs to
\cref{app:stc,app:swmr-proof,app:lo-proof,app:ea}.

\subsection{Single-Threaded Context (STC) Identification}\label{ss:stc}
Accesses that only happen while the program is single-threaded cannot be involved
in any race. In \cref{fig:unified_example}, for instance, function \code{initialize_system} executes before any threads are created, so the write to \code{g_config} at (1) is not instrumented.
\begin{proposition}\label{prop:stc-conc}
    A memory access instruction $x$ that is only ever executed before the program creates additional threads cannot be involved in any data race, and thus $x \notin \memInstrSuff$.
\end{proposition}
We mark potentially multi-threaded code as \textbf{MT}, analyzing whole functions except in $\textsf{main}$, where we distinguish individual basic blocks.
We apply the following rules:
\begin{itemize}
    \item Mark basic blocks of \mainFun\ creating threads.
    \item Mark functions creating threads, as well as functions (or main basic blocks) that may (transitively) call it.
    \item Mark functions that have their address taken.
    \item For any \textbf{MT} function, also mark all functions that it may call.
    \item For any \textbf{MT} basic blocks of \mainFun\, also mark all successors as well as all functions they may call.
\end{itemize}
where \emph{may call} is an overapproximation of program behavior, as computed by the analyzer in a first step.
\begin{tcolorbox}[mygraybox]
Any memory access instruction that is not in a \textbf{MT} function or basic block is guaranteed to only be executed before any further threads are created, and thus cannot be involved in any races and need not be instrumented.
\end{tcolorbox}

\subsection{Single-Writer/Multiple-Reader (SWMR)}
If all writes to a variable occur in single-threaded mode, instrumentation of read accesses is unnecessary,
even if they occur in multi-threaded mode, as two reads cannot race.
In \cref{fig:unified_example}, \code{g_config} is written only in the ST function
\code{initialize_system} (point~1) and read in \code{worker_thread} (point~3);
instrumentation of the read at (3) is thus elided (with summaries, \cref{sec:pipeline_implementation}).
\begin{proposition}\label{prop:swmr}
    An access instruction reading locations $L$ cannot race
    if all instructions potentially writing locations in $L$ can only occur before 
    the program creates further threads.
\end{proposition}

We build this analysis on top of a sound analysis of accessed locations (in our implementation, global variables whose address does not escape), i.e., one that overapproximates
the accessed locations across all executions of our concrete execution model, including 
racy ones (as in \tsan's setting); we do not attempt to reason about language-level undefined behavior.
Let us assume this analysis computes
an overapproximation of access instructions for each memory location,
as well as an overapproximation of memory locations accessed by each access instruction.

\begin{tcolorbox}[mygraybox]
A memory access instruction need not be instrumented if all writes to 
locations potentially accessed by it occur in ST context.
\end{tcolorbox}

\subsection{Lock Ownership (LO) Analysis}
If all accesses to a variable that can happen in multi-threaded mode occur while holding a common lock,
these accesses cannot race.
In \cref{fig:unified_example}, \code{g_counter} is accessed only while holding \code{lock},
so instrumentation of access~(4) is not needed.
\begin{proposition}\label{prop:protection}
    A memory access instruction happening while holding the set of locks $S_0$ that potentially accesses memory locations $L$, cannot be involved in any race if all other instructions potentially accessing locations in $L$
    either happen in ST context or while holding one of the locks in $S_0$.
\end{proposition}

A memory access instruction need not be instrumented if, for all memory locations $\ell$ potentially accessed by it, there is a %
set of locks $S(\ell)$ such that all other instructions potentially accessing $\ell$ in multi-threaded mode also hold
$S(\ell)$.

This set $S(\ell)$ is said to \emph{own} the memory location $\ell$.
An interprocedural, context-insensitive lockset analysis computes the locks guaranteed to be held at each access.
A location's owners are the locks held at every possible access to that location in multi-threaded mode.

\begin{tcolorbox}[mygraybox]
A memory access instruction need not be instrumented if all accessed locations are owned by held locks.
\end{tcolorbox}

\subsection{Escape Analysis (EA)}\label{ss:ea}
Accesses to memory objects that are known to be confined to a single thread cannot be involved in any races. The same holds for fields of such objects.
In \cref{fig:unified_example}, the \code{req_t} object is stack-allocated: our analysis establishes that \code{req->id} remains thread-local (the \code{id} field does not escape, even though \code{payload} does), so access at point~(2a) is not instrumented.
\begin{proposition}\label{prop:ea0}
    A memory access instruction accessing only locations local to its thread cannot
    be involved in any race.
\end{proposition}
An object $o$ is initially considered \textbf{non-escaping} (thread-local). It transitions to the \textbf{escaped} state if a pointer to it, $p$, is used in any of the following ways: %
\begin{enumerate}
\item \textbf{Global Escape:} $p$ is stored into a global variable.
\item \textbf{Transitive Escape:} $p$ is stored into a field of another object $o'$, where $o'$ has escaped.
\item \textbf{Argument Escape:} $p$ is passed as an argument to a function where analysis determines that the corresponding parameter escapes within the callee.
\item \textbf{Return Escape:} $p$ is returned from a function, making it accessible to the caller.
\end{enumerate}
\begin{proposition}\label{prop:ea1}
    An object $o$ that has not escaped according to the definition above remains thread-local,
    and thus any access to it is performed by the thread that created it.
\end{proposition}
Our escape analysis is interprocedural, field-sensitive, and bidirectional, operating in
bottom-up and top-down passes over the call graph (see \cref{fig:ipa_diagram} and \cref{app:ea} for details);
phases are iterated until a fixed point is reached.

\begin{tcolorbox}[mygraybox]
A memory access instruction need not be instrumented if none of the memory locations it potentially accesses
are marked as escaping by the analysis.
\end{tcolorbox}

\begin{figure}
\centering
\scalebox{0.8}{
\begin{tikzpicture}[
    scale=0.9,
    line width=0.65pt, %
    node distance=1.3cm and 1.3cm,
    func_node/.style={
        rectangle, draw, rounded corners, fill=gray!10, thick,
        minimum width=3cm, align=left, inner sep=4pt, font=\sffamily\small
    },
    lib_call/.style={
        func_node, draw=orange!80!black, fill=orange!15, dashed
    },
    annot_box/.style={
        font=\sffamily\footnotesize, align=left, inner sep=3pt,
        draw, rounded corners, fill=white
    },
    stage_annot/.style={
        font=\sffamily\bfseries\small, circle, draw, fill=white, inner sep=1pt, anchor=north east %
    },
    call_edge/.style={->, thick, >=Stealth},
    bottom_up_edge/.style={->, thick, blue, dashed, bend left=30},
    bottom_up_edge2/.style={->, thick, blue, dashed, bend right=30},
    top_down_edge/.style={->, thick, green!40!black, dotted, bend left=30}
]
    \node[func_node] (main) {
        \textbf{\hspace{0.8cm}main} \\
        \texttt{req\_t {\color{blue}o}, {\color{green!50!black}x};} \\
        \texttt{external\_call(\&{\color{green!50!black}x});} \\
        \texttt{func\_A(\&{\color{blue}o}, \&{\color{green!50!black}x});}
    };
    \node[func_node, below=0.6cm of main] (func_A) {
        \textbf{\hspace{0.8cm}func\_A({\color{blue}p}, {\color{green!50!black}y})} \\
        \texttt{func\_B({\color{blue}p});} \\
        \texttt{func\_C({\color{blue}p}, {\color{green!50!black}y});}
    };
    \node[func_node, below left=1.1cm and -0.4cm of func_A] (func_B) {
        \textbf{\hspace{0.6cm}func\_B({\color{blue}q})} \\
        \texttt{external\_call({\color{blue}q})}
    };
    \node[func_node, below right=1.1cm and -0.4cm of func_A] (func_C) {
        \textbf{\hspace{0.6cm}func\_C({\color{blue}r}, {\color{green!50!black}z})} \\
        \texttt{// use {\color{blue}r}, {\color{green!50!black}z}} \code{locally}
    };

    \draw[call_edge] (main) -- (func_A);

    \draw[bottom_up_edge2] (func_B.north east) to (func_A.south);
    \draw[bottom_up_edge] (func_C.north west) to (func_A.south);
    \draw[bottom_up_edge] (func_A.north west) to (main.west);
    \draw[top_down_edge] (main.east) to (func_A.north east);
    \draw[top_down_edge] (func_A.south east) to (func_C.north east);

    \node[annot_box, text=blue, draw=blue, above right=0.1cm and -1.4cm of func_B.north east] (summary_b) {
        \textbf{Summary:} \\ `\color{blue}{q}' is escaping.
    };
    \node[stage_annot] at (summary_b.north west) {1};

    \node[annot_box, text=blue, draw=blue, above left=0.1cm and -2.0cm of func_C.north west] (summary_lib) {
        \textbf{Summary:} \\ `\color{blue}{r}' is not escaping.
    };
    \node[stage_annot] at (summary_lib.north west) {2};

    \node[annot_box, text=blue, draw=blue, below left=-0.4cm and 0.05cm of main.south west] (summary_a) {
        \textbf{Summary:} \\ `\color{blue}{p}' is escaping.
    };
    \node[stage_annot] at (summary_a.north west) {3};

    \node[annot_box, text=green!40!black, draw=green!40!black, below right=-0.4cm and 0.05cm of main.south east] (summary_main_top_down) {
        \textbf{Summary:} \\ `\color{green!50!black}{x}' is escaping.
    };
    \node[stage_annot] at (summary_main_top_down.north east) {4};

    \node[annot_box, text=green!40!black, draw=green!40!black, below right=-1.0cm and 0.6cm of func_A.south east] (summary_func_A_top_down) {
        \textbf{Summary:} \\ `\color{green!50!black}{y}' is escaping.
    };
    \node[stage_annot] at (summary_func_A_top_down.north east) {5};

\end{tikzpicture}}
\caption{Illustration of the bidirectional Interprocedural Escape Analysis. \textbf{\textcircled{1}-\textcircled{3} Bottom-up Pass:} Summaries are generated at the leaves and propagated upwards. \textbf{\textcircled{4}-\textcircled{5} Top-down Pass:} Context from the root is propagated downwards.}
\Description{Call graph in which main calls func\_A, which calls func\_B and func\_C. The bottom-up pass finds that func\_B's parameter q escapes through an external call, that func\_C's parameter r does not escape, and hence that func\_A's parameter p escapes. The top-down pass finds that main passes the address of x to an external call, which makes func\_A's parameter y escape.}
\label{fig:ipa_diagram}
\end{figure}
\section{Beyond Sufficient Instrumentation: Redundancy Elimination}\label{sec:beyond}
Many checks will flag a race if and only if another,
earlier, check on the same memory location also flags a race.
\tsan\ already exploits this to only instrument one access to the
same memory per basic block. Here, we want to go beyond this local notion
of redundancy and also eliminate instrumentation \emph{across} basic blocks.
We once more focus on high-level ideas and defer a more formal treatment 
and proofs of propositions to \cref{app:de-proofs}.

We first define our notion of dominance:
\begin{definition}[Data Race Domination]\label{def:domination}
    A race $r_1$ (involving accesses $a_1$ and $b_1$) is said to (pre-)\emph{dominate} another race $r_2$ (involving $a_2$ and $b_2$) if \emph{every} trace containing $r_2$ also contains $r_1$ and all accesses refer to the same memory locations.
\end{definition}
    An access only involved in dominated races, is
    \emph{redundant}. %
    The set of all redundant accesses is denoted $\memInstrRed$.
\begin{proposition}
    A dynamic data race analyzer that is complete w.r.t.\ instrumented accesses, will, when
    instrumenting at least $\memInstrSuff \setminus \memInstrRed$, report all observed races that are not dominated.
\end{proposition}
Dominated races may still be reported: two accesses retained because each participates
in a non-dominated race may also form a dominated race with each other.

As for $\memInstrSuff$, computing $\memInstrRed$ exactly is not possible.
We instead compute an \emph{underapproximation} of $\memInstrRed$: Removing an
\emph{underapproximation}  of $\memInstrRed$ from an overapproximation of sufficient
instrumentation yields an overapproximation of $\memInstrSuff \setminus\memInstrRed$,
preserving the completeness property. %

\subsection{Underapproximating Redundant Accesses}
To identify redundant accesses, we overapproximate the control-flow by considering all paths in the CFGs,
ignoring guards and demand that all paths between two accesses are free from release-like
synchronization instructions,
ensuring no inter-thread $\to_{hb}$ relationships can be established.
Release-like instructions are those that can appear on the left side of one of the orders other than $\to_{po}$, 
i.e, those that ensure that some other event can only happen later.
For our core language, this is \code{unlock} and \code{create} instructions.

Let $a$ be a memory access instruction. We mark it as redundant if there is a memory access instruction $a_{\text{before}}$ such that the following conditions hold:
\begin{enumerate}[leftmargin=*]
    \item $a_{\text{before}} \text{ dominates } a$; and\label{item:dom}
    \item both access the same (concrete!) memory location\footnote{
        In our implementation, we use the weaker (but still sufficient) requirement that
        the accessed locations \emph{must} alias.
    };\label{item:same-loc}
    \item if $a$ is a write, $a_{\text{before}}$ is also a write;\label{item:write}
    \item all paths from $a_\text{before}$ to $a$ through the \textbf{inter-procedural} CFG are free
    of release-like instructions and external calls.\label{item:no-sync}
\end{enumerate}
\begin{example}
    In the CFG fragment in \cref{fig:dominance_diagram}, 
    \code{BB1} dominates \code{BB3}, instructions \code{I1} and \code{I2} refer to the same memory location,
    both are reads, and the path from \code{I1} to \code{I2} is free from release-like instructions, loops and function calls.
    Thus, instrumentation of \code{I2} is redundant.
\end{example}
\begin{example}
In our running example (\cref{fig:unified_example}), the write \code{g_sum = 0} occurs before the loop. 
This access dominates all operations inside the loop. Consequently, the read of \code{g_sum} within the loop at point (5) (as part of the \code{+=} operation) is dominated and need not be instrumented.
\end{example}

\subsection{Post Dominance}
In \cref{def:domination} (and throughout this paper), we refer to all traces, not only to those
reaching some final state, allowing us to reason about
programs where not all executions terminate regularly.
If one restricts to traces reaching a final state, one can also consider accesses redundant in case
each trace containing a race at some access also contains a race on the same memory location
later in the trace.
Even without this restriction, one can still give a definition
of redundancy that allows more accesses to be eliminated.

\begin{definition}[Data Race Post-Domination]\label{def:post-domination}
    A data race $r_2$ involving accesses $a_2$ and $b_2$ is said to \emph{postdominate} another data race $r_1$ involving accesses $a_1$ and $b_1$ if for every trace $t$ containing $r_1$, two %
    conditions hold: Every trace that has $t$ as a prefix can be extended to a trace $t'$ containing $r_2$; and all accesses refer to the same memory locations.
\end{definition}
Intuitively, this definition requires that, as soon as $r_1$ occurs, $r_2$ must eventually also occur in the same execution.
\begin{definition}[Weakly Sufficient Instrumentation]\label{def:weak-sufficient-instr-set}
    A set of access instructions $\mathcal{W}$ is called \emph{weakly sufficient},
    if for any race $r_1$ (involving instructions $a$ and $b$), either $\{a,b\} \subseteq \mathcal{W}$ or there is a race
    involving instructions $c$ and $d$ such that the race (post-)dominates $r_1$ and
    $\{c,d\} \subseteq \mathcal{W}$.
\end{definition}
Unlike for $\memInstrSuff$ and $\memInstrRed$, there is no unique minimal weakly sufficient instrumentation set.
Consider two threads both writing the same variable twice with no synchronization. It suffices to instrument
one write per thread, yielding four different minimal weakly sufficient instrumentation sets.
\begin{proposition}
    A dynamic data race analyzer that is complete w.r.t.\ instrumented accesses, will, when setting the instrumentation set to
    $\memInstr \supseteq \mathcal{W}$, and not aborting the execution of the program under
    test prematurely, report every observed race --- or an equivalent post- or pre-dominating race.
\end{proposition}
    If $\mathcal{W}$ is weakly sufficient,
    so is the set $\mathcal{W} \cap \memInstrSuff$.

These definitions allow for a more aggressive elimination of redundant accesses including post-dominated accesses. 

Dominance elimination requires the paths between the two accesses to be free of release-like instructions;
post-dominance elimination instead excludes acquire-like instructions.
Acquire-like instructions are those that appear on the right side of one of the orders other than $\to_{po}$, 
i.e, those that can only happen after some other event has happened.
For our core language, this corresponds to \code{lock} instructions.\footnote{
        Thread-join operations, if considered, would also be acquire-like.
    }

We mark a memory access instruction $a$ as being involved in post-dominated races only if there is a memory access instruction 
$a_{\text{after}}$ such that:
\begin{enumerate}[leftmargin=*]
    \item $a_{\text{after}} \text{ post-dominates } a$; and\label{item:dom-post}
    \item both access the same (concrete!) memory location\footnote{
        In our implementation, we use the weaker (but still sufficient) requirement that
        the accessed locations \emph{must} alias.
    };
    \item if $a$ is a write, $a_{\text{after}}$ is also a write;\label{item:write-post}
    \item all paths from $a$ to $a_{\text{after}}$ through the CFG are free
    of acquire-like instructions and external calls;\label{item:no-sync-post}
    \item no path from $a$ to $a_{\text{after}}$ contains loops or function calls.\label{item:no-diverge}
\end{enumerate}
\begin{proposition}
    An instruction $a$ fulfilling (1)-(5) is only involved in post-dominated races.
\end{proposition}
\begin{example}
    Consider the CFG in \cref{fig:dominance_diagram}. 
    \code{BB4} is post-dominated by \code{BB5}, \code{I5} and \code{I3}
    are reads of $z$, and the path from \code{I3} to \code{I5} is free from acquire-like instructions, loops, and  calls.
    Instrumenting \code{I3} is redundant when considering post-dominance.
\end{example}
\begin{figure}[t]
\centering
\scalebox{0.75}{
\begin{tikzpicture}[
    scale=0.75,
    line width=0.7pt, %
    node distance=1cm and 0.5cm,
    bb_node/.style={
        rectangle, draw, rounded corners, fill=gray!10,
        text width=2.8cm, align=center, inner sep=3pt
    },
    annot_text/.style={
        font=\sffamily\footnotesize, align=center,
        draw, inner sep=2pt, fill=white
    },
    conclusion_text/.style={annot_text, fill=red!20}
]
    \node[bb_node] (BB1) {
        \textbf{BB1} (Header) \\
        {\color{red!80!black}\code{I1: read(X);} \texttimes} \\
        \code{if (condition)}
    };
    \node[bb_node, below left=0.7cm and 0.0cm of BB1] (BB3) {
        \textbf{BB3} (then) \\
        {\color{green!50!black}\code{I2: read(X);} \checkmark}
    };
    \node[bb_node, below right=0.7cm and 0.0cm of BB1] (BB4) {
        \textbf{BB4} (else) \\ 
        {\color{green!50!black}\code{I3: read(Z);} \checkmark}
    };
    \node[bb_node, below=2.3cm of BB1] (BB5) {
        \textbf{BB5} (merge) \\
        {\color{red!80!black}\code{I4: read(Y);} \texttimes}
        {\color{red!80!black}\code{I5: read(Z);} \texttimes}
    };

    \draw[->, thick] (BB1.south) -- node[pos=0.4, below=2pt, font=\footnotesize] {\textbf{true}} (BB3.north);
    \draw[->, thick] (BB1.south) -- node[pos=0.4, below=2pt, font=\footnotesize] {\textbf{false}} (BB4.north);
    \draw[->, thick] (BB3.south) -- (BB5.north west);
    \draw[->, thick] (BB4.south) -- (BB5.north east);

    \node[annot_text, fill=green!25, above=0.1cm of BB1] {1. Dominating Access};
    \node[annot_text, fill=green!25, above=0.2cm of BB3] {2. Dominated Access};
    \node[annot_text, fill=green!10, text width=3.3cm, left=0.2cm of BB1, yshift=-0.3cm] {Path from I1 to I2 is free from release-like sync};
    \node[annot_text, fill=blue!15, below=0.1cm of BB5] {3. Post-dominating Access};
    \node[annot_text, fill=blue!15, below=0.05cm of BB4] {4. Post-dominated Access};
    \node[annot_text, fill=green!10, text width=4cm, right=0.2cm of BB5, yshift=0.1cm] {Path from I3 to I5 is free from acquire-like sync, does not contain calls or loops};

    \node[font=\sffamily\footnotesize, below=1.1cm of BB3, align=left] (legend) {
        {\color{red!80!black}\texttimes} Instrumented \\
        {\color{green!50!black}\checkmark} Not Instrumented \\ (Optimized)
    };

\end{tikzpicture}}
\caption{Dominance and Post-dominance. \textbf{Dominance:} Access \texttt{I1} dominates \texttt{I2}. Since the path is free from release-like sync, instrumentation for \texttt{I2} is redundant. \textbf{Post-dominance:} \texttt{BB5} post-dominates both branches, so the instrumented access \texttt{I5} makes checks for \texttt{I3} redundant.}
\Description{Control-flow graph in which BB1 branches to BB3 and BB4, which both lead to BB5. The instrumented read I1 of X in BB1 dominates the read I2 of X in BB3 with no release-like synchronization between them, so I2 is not instrumented. The instrumented read I5 of Z in BB5 post-dominates the read I3 of Z in BB4 with no acquire-like synchronization, call or loop between them, so I3 is not instrumented.}
\label{fig:dominance_diagram}
\end{figure}

\subsection{Trade-off Between Performance and Report Granularity}
(Post-)Dominance-Based Elimination trades report granularity for performance:
since dominated accesses are uninstrumented, race reports point to the (post-)dominating access
(e.g., $I1$ rather than $I2$ in \cref{fig:dominance_diagram}).
\tsan already omits instrumentation for same-location accesses within a single basic block;
our DE is an inter-block generalization of this principle.
This can lead to iterative debugging, but we expect the additional effort to be limited: experienced developers consider the surrounding accesses and synchronization when investigating
a report, and a fix may therefore also address neighboring races.
In our experiments, report relocation was rare: only one report was
relocated (\textsc{memcached}, in one of 30 runs), and the relocated report named
the same function as the original report.

\section{Code Transformations to (Dynamically) Reduce Instrumentation Overhead}\label{sec:trans}
Optimizing compilers iteratively transform input code.
\citet{choi2002efficient} observed that loop unrolling allows marking more accesses as redundant.

\begin{example}
Consider the two code snippets below.
The right snippet could result from the left after loop unrolling, constant propagation, and dead branch elimination.

\begin{center}
\begin{minipage}[t]{0.48\linewidth}
\hspace*{10pt}\textbf{Original}
\begin{minted}[linenos,breaklines=true,xleftmargin=10pt,numbersep=4pt,fontsize=\scriptsize]{C}
int i = 0;
while (i < 10000) {
  g_sum += i; // (A)
  i++;
}
\end{minted}
\end{minipage}\hspace*{-0.03\linewidth}\begin{minipage}[t]{0.55\linewidth}
\textbf{After unrolling/propagation}
\begin{minted}[linenos,breaklines=true,xleftmargin=0pt,numbersep=4pt,fontsize=\scriptsize]{C}
int i = 0;
g_sum += 0; // (B)
i++;
while (i < 10000) {
  g_sum += i; // (C)
  i++;
}
\end{minted}
\end{minipage}
\end{center}

Dominance elimination behaves differently for the two snippets.
While access A is not redundant, access C is dominated by access B and thus is redundant.
Instrumenting the first snippet results in 10000 instrumentation calls, whereas for the second program
only one such call occurs.
\end{example}

We thus follow previous work, and perform some loop unrolling before
applying our dominance-based optimizations.
In the evaluation, we refer to the resulting configuration as \texttt{TSan+DE+Peeling}.

Drawing inspiration from the idea of compiler transformations improving analysis results,
we revisit our \textbf{STC} analysis from \cref{ss:stc}
and augment it with a dynamic counterpart, which unlike loop unrolling is
run \emph{after} the other analyses as part of the instrumentation phase.

\paragraph{Better Detection of Single-Threaded Contexts}

Even when a particular access instruction always happens in single-threaded context,
it needs to be instrumented if the \textbf{STC} analysis cannot establish this.
That can, e.g., happen as \textbf{STC} does not support recovering from a multi-threaded context
after all threads have been joined, which is common for thread pools whose
analysis is challenging for static race detectors~\cite{holter2025sound}.
Beyond this, an access instruction may always happen in single-threaded context for some callers
(namely those that are themselves single-threaded), but not for others.
If there was a way to check whether the program is currently single-threaded, we could
only do the expensive instrumentation call for accesses that happen in a multi-threaded context.
In fact, \tsan already performs some book-keeping when threads are created and joined.
We propose to additionally maintain an atomic counter \code{__thread_cnt} of the number of currently active threads (initialized to $1$).
It is incremented in the instrumentation of thread creation, and decremented in the instrumentation of \code{join} calls.
At access instructions, we then branch on whether \code{__thread_cnt > 1} (i.e., multi-threaded mode).
If it is, we call the instrumentation function, otherwise we skip it.
\begin{example}
Consider the access to \code{g} on the left which was instrumented by \tsan.
Our transformation turns this into the code on the right:
\begin{center}
\begin{minipage}[t]{0.27\linewidth}
\begin{minted}[breaklines=true,numbersep=4pt,fontsize=\scriptsize]{C}
__tsan_write(&g);
g = 5;
\end{minted}
\end{minipage}
\hspace{0.02\linewidth}
\vrule width 0.5pt
\hspace{0.02\linewidth}
\begin{minipage}[t]{0.66\linewidth}
\begin{minted}[breaklines=true,numbersep=4pt,fontsize=\scriptsize]{C}
if(__thread_cnt > 1) __tsan_write(&g);
g = 5;
\end{minted}
\end{minipage}
\end{center}
\end{example}

\begin{theorem}\label{th:stc-dyn}
    The transformation preserves the soundness and completeness of the dynamic race detector, provided
    it disregards synchronization induced by accesses to the counter.
\end{theorem}
A proof sketch is given in \cref{app:stc-proof}.
In our implementation we use the mechanism \tsan uses to maintain its internal state for
the global counter, ensuring the condition from \cref{th:stc-dyn} is satisfied.

While the additional branching does incur some overhead, the branch predictor will often correctly predict 
the taken branch, and the overhead will be greatly outweighed by fewer instrumentation calls.
The goal of almost-zero overhead for this dynamic filter sets us apart from more ambitious approaches like \textsc{ReX}~\cite{huang2017rex} that skip many more instrumentation calls but pay a high per-access cost (see \cref{ss:limitations}).
The same branch also serves a thread that runs while all other threads are parked at known points, such as a server's I/O threads waiting for work; the runtime then marks it as the only active thread.
In the evaluation, we refer to this variant as \texttt{TSan+DynSTC}.
    This transformation does not make the analysis \textbf{STC} redundant.
    Where single-threadedness can be statically proven,
    the overhead of a branch can be avoided and, perhaps more critically, analyses \textbf{SWMR} and \textbf{LO} rely on the
    results of \textbf{STC}. 

\paragraph{Runtime checks.}
On Memcached, \textbf{EA} also skips the checks of objects owned by one thread at a time, such as a connection's buffers, after a runtime check that the current thread is the owner.
On SQLite, \textbf{LO} also skips the checks of objects accessed only under their lock, after a runtime check that the thread holds it; a failed check ends the run.

\paragraph{Inline test.}
\tsan's runtime first tests whether the shadow cell already holds an identical record and returns if so.
We emit this test inline, so that such an access costs a load and a branch instead of a call.

\paragraph{Inline function entry and exit.}
\tsan calls its runtime at every function entry and exit to maintain a call stack.
We inline these calls in the common case and refer to this variant as \texttt{TSan+FE-INL}.

\section{Integration into the \textsc{LLVM} Pipeline}
\label{sec:pipeline_implementation}
We implemented our analyses as \textsc{LLVM} passes that run with a modified instrumentation pass in place of the original.
Like stock TSan, they run at the end of the optimization pipeline, after inlining, loop vectorization and unrolling,
so they analyze the same IR that is finally instrumented (lower half of \cref{fig:filtering_cascade_detailed}).
We also add loop unrolling before instrumentation.

The upper half of \cref{fig:filtering_cascade_detailed} illustrates the sequence of analyses.
While, in general, such a filtering cascade may need to be iterated to reach a fixed point, we have carefully designed
our analyses and their order so that a single pass suffices.
We schedule analyses that provide input for others first. 
First, the \textbf{STC} analysis is run.
Its results are used by the \textbf{SWMR} and \textbf{LO} analyses, as, e.g., accesses to global variables may not be protected by locks during the program's initial single-threaded phase.
\textbf{LO}, \textbf{SWMR}, and \textbf{EA} target different classes of memory and protection mechanisms: \textbf{LO} and \textbf{SWMR} focus on global variables. \textbf{EA} focuses on stack and heap objects. These can be ordered arbitrarily.
In the final stage, the \textbf{DE} pass is executed. It acts as a final clean-up pass, pruning redundant instrumentation.

\begin{remark}
The pipeline can be extended with further analyses beyond those in \cref{sec:analyses}.
Soundness of the overall pipeline follows from soundness of each individual analysis.
\end{remark}
\definecolor{procBlue}{HTML}{4F81BD}
\definecolor{procGreen}{HTML}{9BBB59}
\definecolor{procRed}{HTML}{C0504D}
\definecolor{procGray}{HTML}{555555}

\begin{figure*}
\centering
\resizebox{\textwidth}{!}{%
\begin{tikzpicture}[
  node distance=0.7cm, %
  line width=1.1pt, %
  base_block/.style={
    rectangle, 
    rounded corners=3pt, 
    draw=none,
    text width=3.2cm,      %
    minimum height=1.2cm,  %
    align=center, 
    font=\sffamily\bfseries\fontsize{13}{15}\selectfont
  },
  artifact_block/.style={
    rectangle, 
    rounded corners=3pt, 
    draw=none,
    text width=1.5cm,      %
    minimum height=1.8cm,  %
    align=center, 
    font=\sffamily\bfseries\fontsize{13}{15}\selectfont
  },  artifact_block_wide/.style={
    rectangle, 
    rounded corners=3pt, 
    draw=none,
    text width=3.5cm,      %
    minimum height=2.0cm,  %
    align=center, 
    font=\sffamily\bfseries\fontsize{13}{15}\selectfont
  },
  start_node/.style={base_block, top color=procBlue!20, bottom color=procBlue!40, draw=procBlue},
  filter_node/.style={base_block, top color=white, bottom color=procGreen!20, draw=procGreen},
  unchanged_node/.style={base_block, top color=white, bottom color=procGray!20, draw=procGray},
  end_node/.style={base_block, top color=procRed!10, bottom color=procRed!30, draw=procRed},
  main_arrow/.style={-{Latex[length=3mm, width=2mm]}, line width=1.3pt, draw=procGray},
  dropout_arrow/.style={-{Latex[length=2mm, width=1.5mm]}, dashed, draw=procRed!70},
  dropout_label/.style={
    font=\sffamily\itshape\fontsize{13}{15}\selectfont,
    color=procRed!60!black, 
    align=center,
    below=1mm
  }
]

  \node[artifact_block] (src) {{\Huge\faFileCodeO}\\Code};
  
  \node[unchanged_node, right=of src] (frontend) {Frontend};
  \node[right=of frontend] (dots) {{\footnotesize\ldots}};

  \node[unchanged_node, right=of dots] (mem2reg) {mem2reg};
  \node[unchanged_node, text width=3.8cm, right=of mem2reg] (optim) {Vectorization,\\Inlining, \ldots};

  \node[filter_node, right=of optim] (loopunroll) {Loop\\Unrolling};
  \node[filter_node, right=1.4cm of loopunroll] (instr) {Modified\\Instrumenter};
  \node[right=of instr] (dots2) {{\footnotesize\ldots}};
  \node[unchanged_node, right=of dots2] (asm) {Assembly Generation};
  \node[right=of asm] (dots3) {{\footnotesize\ldots}};
  \node[artifact_block,right=of dots3] (binary) {{\Huge\faGears}\\Binary};

  \coordinate (compilerCenter) at ($(src.west)!0.5!(binary.east)$);
  \node[filter_node] (step3) at ($(compilerCenter)+(0,4.0cm)$) {3. Lock Ownership};
  \node[filter_node, left=of step3] (step2) {2. SWMR};
  \node[filter_node, left=of step2] (step1) {1. STC};
  \node[artifact_block_wide, left=of step1] (start) {All Memory\\Accesses};
  \node[filter_node, right=of step3] (step4) {4. Escape Analysis};
  \node[filter_node, right=of step4] (step5) {5. Dominance Elimination};
  \node[artifact_block_wide, right=of step5] (end) {Accesses to Instrument};

  \draw[main_arrow] (start) -- (step1);
  \draw[main_arrow] (step1) -- (step2);
  \draw[main_arrow] (step2) -- (step3);
  \draw[main_arrow] (step3) -- (step4);
  \draw[main_arrow] (step4) -- (step5);
  \draw[main_arrow] (step5) -- (end);

  \coordinate (mid2) at ($(step1.east)!0.5!(step2.west)$);
  \draw[dropout_arrow] (mid2) -- ++(0, -0.9) node[dropout_label] {STC accesses\\removed};

  \coordinate (mid3) at ($(step2.east)!0.5!(step3.west)$);
  \draw[dropout_arrow] (mid3) -- ++(0, -0.9) node[dropout_label] {SWMR accesses\\removed};

  \coordinate (mid4) at ($(step3.east)!0.5!(step4.west)$);
  \draw[dropout_arrow] (mid4) -- ++(0, -0.9) node[dropout_label] {Lock-protected\\accesses removed};

  \coordinate (mid5) at ($(step4.east)!0.5!(step5.west)$);
  \draw[dropout_arrow] (mid5) -- ++(0, -0.9) node[dropout_label] {Thread-local\\accesses removed};

  \coordinate (mid6) at ($(step5.east)!0.5!(end.west)$);
  \draw[dropout_arrow] (mid6) -- ++(0, -0.9) node[dropout_label] {Redundant\\accesses removed};

  \coordinate (brLeft)  at ($(step1.north west)+(0,0.6)$);
  \coordinate (brRight) at ($(step4.north east)+(0,0.6)$);
  \draw[line width=1.3pt, draw=procGray]
    (brLeft) -- (brRight)
    (brLeft) -- ++(0,-0.25)
    (brRight) -- ++(0,-0.25);
  \node[font=\sffamily\itshape\fontsize{12}{14}\selectfont, text=black]
    at ($(brLeft)!0.5!(brRight)+(0,+0.33)$) {Computing overapproximation of sufficient instrumentation $\memInstrSuff$};

  \coordinate (br2Left)  at ($(step5.north west)+(0,0.6)$);
  \coordinate (br2Right) at ($(step5.north east)+(0,0.6)$);
  \draw[line width=1.3pt, draw=procGray]
    (br2Left) -- (br2Right)
    (br2Left) -- ++(0,-0.25)
    (br2Right) -- ++(0,-0.25);
  \node[font=\sffamily\itshape\fontsize{12}{14}\selectfont, text=black]
    at ($(br2Left)!0.5!(br2Right)+(0,+0.33)$) {Computing $\mathcal{W}$};

  \draw[main_arrow] (loopunroll) -- (instr);
  \draw[main_arrow] (instr) -- (dots2);
  \draw[main_arrow] (dots2) -- (asm);
  \draw[main_arrow] (asm) -- (dots3);
  \draw[main_arrow] (dots3) -- (binary);

  \draw[main_arrow] (src) -- (frontend);
  \draw[main_arrow] (frontend) -- (dots);
  \draw[main_arrow] (dots) -- (mem2reg);
  \draw[main_arrow] (mem2reg) -- (optim);
  \draw[main_arrow] (optim) -- (loopunroll);
  \coordinate (mid7) at ($(loopunroll.east)!0.5!(instr.west)$);
  \node[circle, fill=black, minimum size=6pt, inner sep=0pt] at (mid7) {};
  \coordinate (analysisInputLane) at ($(optim.north)+(0,0.5cm)$);
  \coordinate (analysisOutputLane) at ($(optim.north)+(0,1.0cm)$);
  \draw[main_arrow, rounded corners=4pt]
    (mid7) -- (mid7 |- analysisInputLane)
    -- (start.south |- analysisInputLane) -- (start.south);
  \draw[main_arrow, rounded corners=4pt]
    (end.south) -- (end.south |- analysisOutputLane)
    -- (instr.north |- analysisOutputLane) -- (instr.north);

\end{tikzpicture}
}
\caption{Integration into the \textsc{LLVM} pipeline.
Our analyses are implemented as a sequence of passes that produce a final set of accesses to instrument.
The instrumentation pass is adapted to instrument only for those accesses,
maintain the counter of active threads, and guard instrumentation calls by checks of whether the program is currently multi-threaded.}
\Description{Diagram of the LLVM pipeline: source code, frontend, mem2reg, vectorization and inlining, loop unrolling, the modified instrumenter, assembly generation and the binary. Between loop unrolling and the modified instrumenter, our analyses form a cascade from all memory accesses to the accesses to instrument: STC, SWMR, lock ownership, escape analysis and dominance elimination, each removing the accesses it proves need no check.}
\label{fig:filtering_cascade_detailed}
\end{figure*}

\paragraph{Per-Module Design and Lightweight Whole-Program Analysis}
Our analyses operate per compilation unit, avoiding LTO, whose increased link-time
memory and latency are prohibitive for large codebases such as \textsc{Chromium}.
To recover cross-module precision, our framework supports \emph{summaries}:
after analyzing a module, each analysis serializes its results (e.g., which functions
are single-threaded, which globals are lock-protected, which objects do not escape),
which is used to refine the analysis of subsequent modules.
This provides lightweight whole-program analysis over the program's linked IR, without whole-program optimization or code generation.

\paragraph{Interaction with Other Optimizations}
Eliminating instrumentation creates optimization opportunities because instrumentation calls act as compiler barriers.
Fewer barriers give the scheduler more freedom, may improve instruction-level parallelism and memory-latency hiding,
and reduce register pressure and code size.

\paragraph{Soundness of Dominance Elimination under Instruction Reordering}

For DE to be sound it must be guaranteed that an instrumented dominating access executes before any subsequent (uninstrumented) dominated access
--- also at \emph{runtime}.
Instruction reordering, by the compiler or hardware, may appear to challenge
this.
We argue that our approach remains sound in the face of both.

\begin{wrapfigure}{r}{0.45\columnwidth}
\vspace{-1.2em}
\centering
\footnotesize

\begin{minipage}{\linewidth}
\begin{minted}[breaklines=true,fontsize=\scriptsize]{c}
// BB1: Dominator block
__tsan_write(&x);
x = 1; // Dominating access a1

if (condition)
  // BB3: Dominated block
  x = 2; // Dominated access a2
\end{minted}
\end{minipage}

\vspace{-0.8em}
\captionsetup{justification=raggedright, singlelinecheck=false} %
\caption{Example analyzing reordering effects on Dominance-based Analysis.}
\Description{Code listing: block BB1 writes x = 1 after an instrumentation call (access a1), and block BB3, reached under a condition, writes x = 2 without one (access a2), which a1 dominates.}
\label{fig:de_reordering_example}
\vspace{-1.2em}
\end{wrapfigure}
\cref{fig:de_reordering_example} shows two accesses to \code{x}, with $a_1$ dominating $a_2$. Our approach instruments $a_1$ but not $a_2$. The question is whether reordering could cause  $a_2$ to execute before $a_1$.
An optimization might attempt to lift \code{x = 2} out of the \code{if}-block, but the
preceding instrumentation call acts as a compiler barrier, preventing memory
accesses from being moved across it.
Thus, compiler-driven reordering is harmless.
At the hardware level, our analysis, like \tsan, reasons about \textit{architectural state}.
Processors commit instructions in program order, so speculative execution of $a_2$ cannot commit before $a_1$.
Both arguments also apply to post-dominance.

\paragraph{Effects of Bounded Shadow Memory}
Our completeness results hold for an idealized version of \tsan with unbounded shadow state.
\tsan uses \emph{bounded} shadow state (four cells per granule) and evicts pseudo-randomly
instead of by age (FIFO).
Hence any change to the instrumentation --- including stock \tsan's intra-basic-block elision --- potentially changes which records are evicted and thus which races are reported.
Analyses other than \textbf{DE} remove instrumentation only for accesses that cannot race, but their removal can still change subsequent eviction victims.
\textbf{DE} additionally removes redundant accesses. If a dominating access's record is evicted before a conflicting access occurs, not instrumenting the dominated access can lose a report that stock \tsan would produce. Conversely, recording the dominated access can itself evict another racy access, which our \textbf{DE} may therefore preserve.

We did not observe such an effect: per-granule eviction counters show the granules carrying the reported races in \textsc{SQLite} and \textsc{memcached} evicted at the same rate (within 3--7\%) under stock \tsan and the optimized builds.

\section{Experimental Evaluation}
\label{sec:eval}

Our evaluation first quantifies by how much overhead is reduced and identifies key drivers of this reduction for different workloads (\RQ{1}).
We measure compilation overhead and empirically confirm that \tsan's race detection is preserved (\RQ{2}).
Finally, we explore the limits of static analysis by identifying challenging code patterns and compare our approach to a purely dynamic filtering technique to evaluate the trade-offs in redundancy elimination and performance (\RQ{3}). 

\paragraph{Experimental Methodology.}
To evaluate our approach, we selected a diverse suite of concurrent real-world applications: the \textbf{Memcached} server, the \textbf{Redis} in-memory store (evaluated with `redis-benchmark'), the \textbf{FFmpeg} multimedia framework, the \textbf{MySQL} database (with `sysbench'), the \textbf{SQLite} embedded database (with `threadtest3'), and the \textbf{Chromium} browser, evaluating our optimizations on workloads ranging from I/O-bound servers to CPU-intensive desktop applications.

Experiments ran on a dual-socket AMD EPYC 9115 (32 cores; Chromium on a 56-core Intel Xeon w9-3495X) with Ubuntu 24.04 LTS, using
\textsc{LLVM/Clang} 19 at \code{-O2}.
Except for Chromium, every \tsan configuration was built four times, differing only in code placement, and run four times per build (Native: one build, 16 runs).
Speedups are ratios of geometric means with 95\% bootstrap intervals.
We compare four configurations: \textbf{Native} (uninstrumented); \textbf{Original TSan}; \textbf{TSan+<Opt>} (individual analyses enabled); and \textbf{TSan+AllOpt} which encompasses \textbf{EA}, \textbf{LO}, \textbf{STC}, \textbf{SWMR}, and \textbf{DE}.
We also show three variants from \cref{sec:trans}: \textbf{TSan+DE+Peeling}, \textbf{TSan+DynSTC}, and \textbf{TSan+FE-INL}.
To study the impact of parallelism, we vary the number of active threads for a subset of benchmarks.
We use the metrics in \cref{fig:metric_defsum} with
indices denoting configurations. %
\begin{figure*}
\centering
\begingroup
\small
\setlength{\tabcolsep}{4pt}
\renewcommand{\arraystretch}{1.05}
\resizebox{\textwidth}{!}{%
\begin{tabular}{@{}c@{}}
\begin{tabular}{@{}r@{\;:\;}l@{\quad}r@{\;:\;}l@{\qquad}r@{\;:\;}l@{}}
\multicolumn{4}{@{}l}{\textit{\textbf{Base Quantities}}} &
\multicolumn{2}{l@{}}{\textit{\textbf{Configuration Subscripts}}} \\
$T$ & Operation time &
$Q$ & \# Instrumented memory-access instructions &
$\text{nat}$ & Native build (no \tsan) \\
$F$ & Peak memory consumption &
$L$ & \# Executed memory instrumentation calls &
$\text{orig}$ & Original \tsan \\
\multicolumn{2}{@{}l@{\quad}}{} &
$M$ & \# Instrumented memory locations at runtime &
$\text{opt}$ & Optimized \tsan (e.g., \code{TSan+EA}, \code{TSan+AllOpt}) \\
\multicolumn{4}{@{}l}{} 
\end{tabular}
\\[2em]
\begin{tabular}{@{}ll@{\qquad}ll@{\qquad}ll@{}}
\multicolumn{2}{@{}l}{\textit{\textbf{Runtime Performance}}} &
\multicolumn{2}{l}{\textit{\textbf{Memory Consumption}}} &
\multicolumn{2}{l@{}}{\textit{\textbf{Instrumentation Reduction}}} \\
Slowdown vs. Native (SD) & $T_{\text{opt}} / T_{\text{nat}}$ &
Increase Factor (IF) & $F_{\text{opt}} / F_{\text{nat}}$ &
Static (SIR) & $(Q_{\text{orig}} - Q_{\text{opt}})\,/\,Q_{\text{orig}}$ \\
Speedup vs. TSan (SU) & $T_{\text{orig}} / T_{\text{opt}}$ &
Memory Ratio vs. TSan (MR) & $F_{\text{opt}} / F_{\text{orig}}$ &
Dynamic (DIR) & $(L_{\text{orig}} - L_{\text{opt}})\,/\,L_{\text{orig}}$ \\
\multicolumn{2}{@{}l}{\textit{\textbf{Compilation Overhead}}} &
& &
Dynamic Memory Locations (DMIR) & $(M_{\text{orig}} - M_{\text{opt}})\,/\,M_{\text{orig}}$ \\
Build time (CTO) & $T^{\text{build}}_{\text{opt}} - T^{\text{build}}_{\text{orig}}$ & & & & \\
\end{tabular}
\end{tabular}}
\endgroup
\caption{Definitions of Base Quantities and Configuration Subscripts and summary of metrics.}
\Description{Table of the base quantities (operation time, peak memory, instrumented memory-access instructions, executed instrumentation calls, instrumented memory locations), the configuration subscripts (native, Original TSan, optimized), and the metrics computed from them for runtime performance, memory consumption, instrumentation reduction and build time.}
\label{fig:metric_defsum}
\end{figure*}
\subsection{Overhead Reduction and Instrumentation Elimination (\RQ{1})}
\label{ssec:results}
We report the baseline overhead of the original \tsan to motivate the need for optimization, and then study the impact of our optimizations.

\paragraph*{Baseline Overhead}

We measured the slowdown of the original \tsan relative to a native, uninstrumented build (\cref{fig:tsan_slowdown_native}).
The overhead varies widely, from about \(2.8\times\) (FFmpeg) to about \(10\times\) (MySQL). SQLite ($5.4\times$), Redis ($6.2\times$), and Memcached ($8.5\times$) lie in between.
Within FFmpeg, the slowdown ranges from 1.3\(\times\) (H.264 encode) to 6.3\(\times\) (MJPEG encode).
In Chromium, the Speedometer~3.0 score drops from 2.39 to 0.25, a $9.6\times$ slowdown.

\begin{figure}
    \centering
    \pgfplotsset{
        tsanSlowdownBase/.style={
            ybar,
            height=3.1cm,
            xtick=data,
            xticklabel style={
                font=\tiny
            },
            tick style={draw=none},
            ylabel style={font=\scriptsize, align=center},
            label style={font=\small},
            tick label style={font=\scriptsize},
            nodes near coords,
            nodes near coords style={
                font=\tiny,
                /pgf/number format/fixed,
                /pgf/number format/fixed zerofill,
                /pgf/number format/precision=1,
                anchor=south,
                yshift=0.4pt
            },
            clip=false,
        },
        tsanSlowdownChrome/.style={
            tsanSlowdownBase,
            width=\linewidth,
            bar width=0.30cm,
            symbolic x coords={Chromium},
            ymin=0, ymax=14.0, ybar=0.15cm,
            ytick distance=2,
            ylabel={\scriptsize Slowdown vs. Native},
            ylabel near ticks,
            legend style={
                font=\scriptsize,
                draw=none,
                fill=none,
                legend columns=2,
                column sep=0.5em,
                cells={align=left}
            },
            legend image code/.code={
                \draw[#1] (0cm,-0.10cm) rectangle (0.14cm,0.16cm);
            },
            legend to name=tsanSlowdownLegend,
            enlarge x limits=0.45
        },
        tsanSlowdownOthers/.style={
            tsanSlowdownBase,
            width=\linewidth,
            bar width=0.28cm,
            symbolic x coords={Redis,MySQL,FFmpeg,Memcached,SQLite},
            ymin=0, ymax=14.0, ybar=3pt,
            ytick distance=2,
            enlarge x limits=0.11
        }
    }

    \begin{minipage}[t]{0.30\linewidth}
        \centering
        \begin{tikzpicture}[baseline=(current axis.south)]
            \begin{axis}[tsanSlowdownChrome, width=\linewidth]
                \addplot[style={fill=blue!70, draw=none}, error bars/.cd, y dir=both, y explicit]
                    coordinates {
                        (Chromium,9.554) += (0,0.349) -= (0,0.300)
                    };
                \addplot[style={fill=red!56, draw=none}, error bars/.cd, y dir=both, y explicit]
                    coordinates {
                        (Chromium,8.359) += (0,0.496) -= (0,0.455)
                    };
                \legend{Original \tsan, \texttt{TSan+AllOpt}}
            \end{axis}
        \end{tikzpicture}
    \end{minipage}
    \hspace*{-0.110\linewidth}
    \begin{minipage}[t]{0.79\linewidth}
        \centering
        \begin{tikzpicture}[baseline=(current axis.south)]
            \begin{axis}[tsanSlowdownOthers, width=\linewidth]
                \addplot[style={fill=blue!70, draw=none}, error bars/.cd, y dir=both, y explicit]
                    coordinates {
                        (Redis,6.212) += (0,0.038) -= (0,0.069)
                        (MySQL,10.000) += (0,0.270) -= (0,0.260)
                        (FFmpeg,2.762) += (0,0.025) -= (0,0.015)
                        (Memcached,8.475) += (0,0.109) -= (0,0.087)
                        (SQLite,5.393) += (0,0.217) -= (0,0.129)
                    };
                \addplot[style={fill=red!56, draw=none}, error bars/.cd, y dir=both, y explicit]
                    coordinates {
                        (Redis,4.222) += (0,0.062) -= (0,0.062)
                        (MySQL,9.090) += (0,0.260) -= (0,0.240)
                        (FFmpeg,2.170) += (0,0.019) -= (0,0.012)
                        (Memcached,4.443) += (0,0.060) -= (0,0.054)
                        (SQLite,4.111) += (0,0.136) -= (0,0.107)
                    };
            \end{axis}
        \end{tikzpicture}
    \end{minipage}

    \begin{center}
        \pgfplotslegendfromname{tsanSlowdownLegend}
    \end{center}

    \caption{Slowdown of original \tsan and of \texttt{TSan+AllOpt} compared to native execution. For Chromium, derived from Speedometer~3.0 scores; for the others, from time or throughput metrics.
    }
    \Description{Bar charts of the slowdown over native execution of Original TSan and of TSan+AllOpt, one pair of bars per application, on one scale: Chromium in a panel on the left, and Redis, MySQL, FFmpeg, Memcached and SQLite on the right. On every application the TSan+AllOpt bar is the lower one.}
    \label{fig:tsan_slowdown_native}
\end{figure}

\paragraph*{Runtime Performance with Optimizations}

\code{TSan+}\allowbreak\code{AllOpt} is faster than Original \tsan on every application (see \cref{fig:all_results}). In each application but Redis and Chromium, one optimization dominates. Interestingly, which one differs by application: \textbf{EA} on Memcached, \textbf{STC} on FFmpeg, \textbf{LO} on SQLite, and \textbf{FE-INL} on MySQL.
Almost no optimization slows an application down beyond measurement noise, and \code{TSan+AllOpt} is never measurably slower than the best single configuration. The clearest exception, EA on Speedometer (\cref{fig:chromium_selected_speedometer}), probably reflects code layout: it removes only 0.8\% of the instrumentation sites (\cref{fig:sir_reduction}).

\textbf{Memcached.}
This highly optimized server application shows a speedup of $1.91\times$ with our full optimization suite. EA alone gives $1.84\times$ (\cref{fig:results_memcached}).

\textbf{Redis.}
Redis, involving many small CPU-bound operations, benefits significantly, showing an average speedup of $1.47\times$ with \code{TSan+AllOpt}. The optimizations compound here: DynSTC marks the phases in which a check may be skipped, and the inline test skips it without a runtime call, so \code{TSan+AllOpt} avoids 91\% of runtime calls (\cref{fig:results_redis,fig:dir_reduction}).

\textbf{FFmpeg.}
This CPU-intensive multimedia framework also shows substantial speedups (\cref{fig:results_ffmpeg}), achieving a $1.27\times$ average speedup with \code{TSan+AllOpt}. STC, which here includes DynSTC and the inline test, gives the same $1.27\times$.

\textbf{MySQL.}
For this complex DBMS, improvements are more modest but consistent. \code{TSan+AllOpt} achieves a $1.08\times$ speedup on the `Select' benchmark and $1.10\times$ on the write scripts (insert, update, delete) (\cref{fig:results_mysql_select,fig:results_mysql_write}). \code{TSan+FE-INL} is the most effective variant, providing speedups of $1.06\times$ and $1.09\times$, respectively.

\textbf{SQLite}
reaches $1.31\times$ with \code{TSan+AllOpt} (\cref{fig:results_sqlite}), and LO alone reaches a similar speedup.

\textbf{Chromium.}
To evaluate our approach on a large C++ desktop application, we used a suite including Speedometer 3.0 and Blink benchmarks. Across the 48 Blink benchmarks, the geometric mean speedup is $1.25\times$ (\cref{fig:chromium_aggregated}). At the suite level, medians range from $1.04\times$ for Image Decoder to $1.29\times$ for Layout. Representative per-benchmark comparisons are shown in \cref{fig:results_chromium_parser,fig:results_chromium_paint}.
This confirms our approach is effective on large, real-world applications.

\begin{figure*}[t]
    \centering
    \pgfplotsset{
        allbenchstyle/.style={
            ybar,
            yticklabel style={/pgf/number format/fixed, /pgf/number format/fixed zerofill, /pgf/number format/precision=1},
            width=1.15\linewidth,
            height=3.5cm,
            bar width=0.22cm,
            bar shift=0pt,
            xtick={1,2,3,4,5,6,7,8},
            xticklabels={EA, LO, STC, SWMR, DE, DE+Peeling, FE-INL, AllOpt},
            xticklabel style={
                font=\tiny, 
                anchor=east, 
                rotate=30, 
                align=right,
                inner sep=1pt
            },
            tick style={draw=none},
            ymin=0.9,
            ytick distance=0.1,
            ylabel style={font=\scriptsize},
            label style={font=\tiny},
            tick label style={font=\tiny},
            nodes near coords,
            nodes near coords style={
                font=\tiny,
                /pgf/number format/fixed,
                /pgf/number format/precision=2,
                anchor=south
            },
            clip=false,
            enlarge x limits={abs=0.45}
        }
    }

    \subfloat[Memcached\label{fig:results_memcached}]{%
        \begin{minipage}{0.245\textwidth}\centering
        \begin{tikzpicture}
            \begin{axis}[allbenchstyle, ymax=2.3, ytick distance=0.2,
                nodes near coords style={rotate=90, anchor=west, inner sep=2.5pt},
                xtick={1,2,3,4,5,6,7,8},
                xticklabels={EA, LO, STC, SWMR, DE, DE+Peeling, FE-INL, AllOpt}]
            \addplot[style={fill=blue!70, draw=none}, nodes near coords={}, error bars/.cd, y dir=both, y explicit]
                coordinates {(1,1.844) += (0,0.009) -= (0,0.008) (2,0.999) += (0,0.004) -= (0,0.004) (3,0.999) += (0,0.003) -= (0,0.004) (4,1.013) += (0,0.003) -= (0,0.003) (5,1.004) += (0,0.003) -= (0,0.002) (6,1.005) += (0,0.002) -= (0,0.003) (7,1.001) += (0,0.003) -= (0,0.002)};
            \addplot[style={fill=red!56, draw=none}, nodes near coords={}, error bars/.cd, y dir=both, y explicit]
                coordinates {(8,1.907) += (0,0.010) -= (0,0.007)};
            \addplot[draw=none, fill=none, point meta=explicit symbolic]
                coordinates {(1,1.853) [1.84] (2,1.003) [1.00] (3,1.002) [1.00] (4,1.016) [1.01] (5,1.007) [1.00] (6,1.007) [1.00] (7,1.004) [1.00] (8,1.917) [1.91]};
            \end{axis}
        \end{tikzpicture}
        \end{minipage}%
    }%
    \subfloat[Redis\label{fig:results_redis}]{%
        \begin{minipage}{0.245\textwidth}\centering
        \begin{tikzpicture}
            \begin{axis}[allbenchstyle, ymax=1.7,
                nodes near coords style={rotate=90, anchor=west, inner sep=2.5pt},
                xtick={1,2,3,4,5,6,7,8},
                xticklabels={EA, LO, STC, SWMR, DE, DE+Peeling, FE-INL, AllOpt}]
            \addplot[style={fill=blue!70, draw=none}, nodes near coords={}, error bars/.cd, y dir=both, y explicit]
                coordinates {(1,1.007) += (0,0.009) -= (0,0.007) (2,0.999) += (0,0.009) -= (0,0.006) (3,1.159) += (0,0.014) -= (0,0.010) (4,1.000) += (0,0.011) -= (0,0.009) (5,1.095) += (0,0.009) -= (0,0.009) (6,1.091) += (0,0.011) -= (0,0.012) (7,1.061) += (0,0.015) -= (0,0.018)};
            \addplot[style={fill=red!56, draw=none}, nodes near coords={}, error bars/.cd, y dir=both, y explicit]
                coordinates {(8,1.471) += (0,0.017) -= (0,0.024)};
            \addplot[draw=none, fill=none, point meta=explicit symbolic]
                coordinates {(1,1.016) [1.01] (2,1.008) [1.00] (3,1.173) [1.16] (4,1.011) [1.00] (5,1.104) [1.10] (6,1.102) [1.09] (7,1.076) [1.06] (8,1.488) [1.47]};
            \end{axis}
        \end{tikzpicture}
        \end{minipage}%
    }%
    \subfloat[FFmpeg\label{fig:results_ffmpeg}]{%
        \begin{minipage}{0.245\textwidth}\centering
        \begin{tikzpicture}
            \begin{axis}[allbenchstyle, ymax=1.45,
                nodes near coords style={rotate=90, anchor=west, inner sep=2.5pt},
                xtick={1,2,3,4,5,6,7,8},
                xticklabels={EA, LO, STC, SWMR, DE, DE+Peeling, FE-INL, AllOpt}]
            \addplot[style={fill=blue!70, draw=none}, nodes near coords={}, error bars/.cd, y dir=both, y explicit]
                coordinates {(1,1.001) += (0,0.032) -= (0,0.009) (2,0.999) += (0,0.032) -= (0,0.010) (3,1.274) += (0,0.005) -= (0,0.005) (4,0.997) += (0,0.028) -= (0,0.019) (5,1.023) += (0,0.031) -= (0,0.009) (6,1.015) += (0,0.028) -= (0,0.014) (7,1.004) += (0,0.028) -= (0,0.014)};
            \addplot[style={fill=red!56, draw=none}, nodes near coords={}, error bars/.cd, y dir=both, y explicit]
                coordinates {(8,1.272) += (0,0.008) -= (0,0.005)};
            \addplot[draw=none, fill=none, point meta=explicit symbolic]
                coordinates {(1,1.033) [1.00] (2,1.031) [1.00] (3,1.279) [1.27] (4,1.025) [1.00] (5,1.054) [1.02] (6,1.043) [1.01] (7,1.032) [1.00] (8,1.280) [1.27]};
            \end{axis}
        \end{tikzpicture}
        \end{minipage}%
    }%
    \subfloat[SQLite\label{fig:results_sqlite}]{%
        \begin{minipage}{0.245\textwidth}\centering
        \begin{tikzpicture}
            \begin{axis}[allbenchstyle, ymax=1.60,
                nodes near coords style={rotate=90, anchor=west, inner sep=2.5pt},
                xtick={1,2,3,4,5,6,7,8},
                xticklabels={EA, LO, STC, SWMR, DE, DE+Peeling, FE-INL, AllOpt}]
            \addplot[style={fill=blue!70, draw=none}, nodes near coords={}, error bars/.cd, y dir=both, y explicit]
                coordinates {(1,0.993) += (0,0.036) -= (0,0.026) (2,1.322) += (0,0.039) -= (0,0.027) (3,0.995) += (0,0.030) -= (0,0.026) (4,0.999) += (0,0.027) -= (0,0.023) (5,0.992) += (0,0.025) -= (0,0.027) (6,1.005) += (0,0.031) -= (0,0.021) (7,0.984) += (0,0.039) -= (0,0.034)};
            \addplot[style={fill=red!56, draw=none}, nodes near coords={}, error bars/.cd, y dir=both, y explicit]
                coordinates {(8,1.312) += (0,0.054) -= (0,0.034)};
            \addplot[draw=none, fill=none, point meta=explicit symbolic]
                coordinates {(1,1.029) [0.99] (2,1.361) [1.32] (3,1.025) [0.99] (4,1.026) [1.00] (5,1.017) [0.99] (6,1.036) [1.00] (7,1.023) [0.98] (8,1.366) [1.31]};
            \end{axis}
        \end{tikzpicture}
        \end{minipage}%
    }%

    \caption{Runtime Performance (Average Speedup vs. Original TSan). Higher is better.}
    \Description{Four bar charts, one each for Memcached, Redis, FFmpeg and SQLite, of the speedup over Original TSan of each optimization alone and of AllOpt, with error bars. AllOpt is above 1 on all four applications.}
    \label{fig:all_results}
\end{figure*}

\begin{figure*}[t]
    \centering
    \captionsetup[subfloat]{justification=centering, singlelinecheck=false}

    \pgfplotsset{
        mysqlstyle/.style={
            ybar,
            yticklabel style={/pgf/number format/fixed, /pgf/number format/fixed zerofill, /pgf/number format/precision=1},
            width=1.2\linewidth,
            height=3.5cm,
            bar width=0.22cm,
            bar shift=0pt,
            xtick={1,2,3,4,5,6,7,8},
            xticklabels={EA, LO, STC, SWMR, DE, DE+Peeling, FE-INL, AllOpt},
            xticklabel style={font=\tiny, anchor=east, rotate=30, align=right},
            tick style={draw=none}, %
            ymin=0.9, 
            ymax=1.70,
            ytick distance=0.1,
            ylabel near ticks, 
            ylabel style={font=\tiny},
            label style={font=\tiny},
            tick label style={font=\tiny},
            title style={font=\small\bfseries, yshift=-1ex},
            nodes near coords,
            nodes near coords style={font=\tiny, /pgf/number format/fixed, /pgf/number format/fixed zerofill, /pgf/number format/precision=2, anchor=south, inner sep=2.5pt},
            clip=false,             %
            enlarge x limits={abs=0.5},  %
        }
    }

    \subfloat[MySQL\\\texttt{select-random-points}\label{fig:results_mysql_select}]{
        \begin{minipage}{0.24\textwidth}
        \centering %
        \begin{tikzpicture}
            \begin{axis}[
                mysqlstyle,
                ymax=1.12,
                ytick distance=0.1,
                xtick={1,2,3,4,5,6,7,8},
                xticklabels={EA, LO, STC, SWMR, DE, DE+Peeling, FE-INL, AllOpt}]
            \addplot[style={fill=blue!70, draw=none}, nodes near coords={}, error bars/.cd, y dir=both, y explicit]
                coordinates {(1,1.007) += (0,0.010) -= (0,0.012) (2,1.004) += (0,0.009) -= (0,0.029) (3,1.001) += (0,0.012) -= (0,0.015) (4,1.002) += (0,0.007) -= (0,0.015) (5,1.005) += (0,0.009) -= (0,0.013) (6,1.003) += (0,0.009) -= (0,0.007) (7,1.064) += (0,0.010) -= (0,0.010)};
            \addplot[style={fill=red!56, draw=none}, nodes near coords={}, error bars/.cd, y dir=both, y explicit]
                coordinates {(8,1.077) += (0,0.008) -= (0,0.011)};
            \addplot[draw=none, fill=none, point meta=explicit symbolic]
                coordinates {(1,1.017) [1.01] (2,1.013) [1.00] (3,1.013) [1.00] (4,1.009) [1.00] (5,1.014) [1.01] (6,1.012) [1.00] (7,1.073) [1.06] (8,1.085) [1.08]};
            \end{axis}
        \end{tikzpicture}
        \end{minipage}
    }
    \subfloat[MySQL\\\texttt{insert, update, delete}\label{fig:results_mysql_write}]{
        \begin{minipage}{0.24\textwidth}
        \centering
        \begin{tikzpicture}
            \begin{axis}[
                mysqlstyle,
                ymax=1.16,
                ytick distance=0.1,
                xtick={1,2,3,4,5,6,7,8},
                xticklabels={EA, LO, STC, SWMR, DE, DE+Peeling, FE-INL, AllOpt}]
            \addplot[style={fill=blue!70, draw=none}, nodes near coords={}, error bars/.cd, y dir=both, y explicit]
                coordinates {(1,1.007) += (0,0.027) -= (0,0.007) (2,1.002) += (0,0.025) -= (0,0.005) (3,1.003) += (0,0.022) -= (0,0.010) (4,1.004) += (0,0.025) -= (0,0.006) (5,1.004) += (0,0.017) -= (0,0.008) (6,1.001) += (0,0.017) -= (0,0.008) (7,1.089) += (0,0.020) -= (0,0.007)};
            \addplot[style={fill=red!56, draw=none}, nodes near coords={}, error bars/.cd, y dir=both, y explicit]
                coordinates {(8,1.100) += (0,0.004) -= (0,0.004)};
            \addplot[draw=none, fill=none, point meta=explicit symbolic]
                coordinates {(1,1.034) [1.01] (2,1.027) [1.00] (3,1.025) [1.00] (4,1.029) [1.00] (5,1.021) [1.00] (6,1.018) [1.00] (7,1.109) [1.09] (8,1.104) [1.10]};
            \end{axis}
        \end{tikzpicture}
        \end{minipage}
    }
    \subfloat[Chromium Parser\\\texttt{HTML5 Render}\label{fig:results_chromium_parser}]{
        \begin{minipage}{0.24\textwidth}
        \centering
        \begin{tikzpicture}
            \begin{axis}[mysqlstyle, ymax=1.3, xtick={1,2,3,4,5,6,7,8}, xticklabels={EA, LO, STC, SWMR, DE, DE+Peeling, FE-INL, AllOpt}]
            \addplot[style={fill=blue!70, draw=none}, nodes near coords={}, error bars/.cd, y dir=both, y explicit]
                coordinates {(1,0.997) += (0,0.005) -= (0,0.005) (2,0.993) += (0,0.012) -= (0,0.012) (3,0.993) += (0,0.012) -= (0,0.012) (4,0.993) += (0,0.012) -= (0,0.012) (5,0.995) += (0,0.011) -= (0,0.010) (6,0.990) += (0,0.004) -= (0,0.004) (7,1.095) += (0,0.009) -= (0,0.008)};
            \addplot[style={fill=red!56, draw=none}, nodes near coords={}, error bars/.cd, y dir=both, y explicit]
                coordinates {(8,1.208) += (0,0.014) -= (0,0.014)};
            \addplot[draw=none, fill=none, point meta=explicit symbolic]
                coordinates {(1,1.002) [1.00] (2,1.005) [0.99] (3,1.005) [0.99] (4,1.005) [0.99] (5,1.006) [0.99] (6,0.994) [0.99] (7,1.104) [1.09] (8,1.222) [1.21]};
            \end{axis}
        \end{tikzpicture}
        \end{minipage}
    }
    \subfloat[Chromium Paint\\\texttt{Transform Changes}\label{fig:results_chromium_paint}]{
        \begin{minipage}{0.24\textwidth}
        \centering
        \begin{tikzpicture}
            \begin{axis}[mysqlstyle, ymax=1.5, xtick={1,2,3,4,5,6,7,8}, xticklabels={EA, LO, STC, SWMR, DE, DE+Peeling, FE-INL, AllOpt}]
            \addplot[style={fill=blue!70, draw=none}, nodes near coords={}, error bars/.cd, y dir=both, y explicit]
                coordinates {(1,0.999) += (0,0.054) -= (0,0.056) (2,1.000) += (0,0.049) -= (0,0.046) (3,1.000) += (0,0.049) -= (0,0.046) (4,1.000) += (0,0.049) -= (0,0.046) (5,1.004) += (0,0.052) -= (0,0.049) (6,1.002) += (0,0.047) -= (0,0.040) (7,1.119) += (0,0.038) -= (0,0.020)};
            \addplot[style={fill=red!56, draw=none}, nodes near coords={}, error bars/.cd, y dir=both, y explicit]
                coordinates {(8,1.381) += (0,0.003) -= (0,0.003)};
            \addplot[draw=none, fill=none, point meta=explicit symbolic]
                coordinates {(1,1.053) [1.00] (2,1.049) [1.00] (3,1.049) [1.00] (4,1.049) [1.00] (5,1.056) [1.00] (6,1.049) [1.00] (7,1.157) [1.12] (8,1.384) [1.38]};
            \end{axis}
        \end{tikzpicture}
        \end{minipage}
    }

    \caption{MySQL and selected Chromium Runtime Performance (Average Speedup vs. Original TSan). Higher is better. Each panel shows a specific benchmark test case (dataset) within MySQL or Chromium.}
    \Description{Four bar charts comparing speedup versus original TSan. The first two charts show MySQL Select and the write scripts, where AllOpt reaches 1.08x and 1.10x. The third chart shows Chromium Parser HTML5 Render, where AllOpt reaches 1.21x. The fourth chart shows Chromium Paint Transform Changes, where AllOpt reaches 1.38x.}
    \label{fig:mysql_speedup}

\end{figure*}

\begin{figure}[t!]
    \centering

    \pgfplotsset{
        chromestyle/.style={
            yticklabel style={/pgf/number format/fixed, /pgf/number format/fixed zerofill, /pgf/number format/precision=1},
            width=0.9\linewidth,
            height=3.95cm,
            tick label style={font=\footnotesize},
            label style={font=\scriptsize},
            legend style={font=\scriptsize},
            title style={font=\small, yshift=-0.5ex},
            trim axis left, trim axis right,
            anchor=north
        }
    }

    \begin{tikzpicture}
        \begin{axis}[
            chromestyle,
            boxplot/draw direction=y,
            ylabel style={font=\small, align=center},
            ylabel={{\small Speedup (AllOpt)}},
            ylabel near ticks,
            xtick={1,2,3,4,5,6,7},
            xticklabels={Speedometer, Parser, DOM, SVG, Layout, Paint, ImgDec},
            xticklabel style={
                font=\scriptsize,
                rotate=20,
                anchor=north east,
                align=right,
                inner sep=1pt
            },
            ymin=0.8, ymax=2.0,
            ytick distance=0.2,
            enlarge x limits=0.05,
            every axis plot/.append style={line width=0.5pt},
            boxplot/every box/.style={fill=blue!30, draw=blue!80!black, thick},
            boxplot/every whisker/.style={draw=blue!80!black, thick},
            boxplot/every median/.style={draw=red, thick},
        ]

        \addplot[boxplot prepared={
            lower whisker=1.06, lower quartile=1.08, median=1.14, upper quartile=1.20, upper whisker=1.27
        }, draw=black] coordinates {};

        \addplot[boxplot prepared={
            lower whisker=1.00, lower quartile=1.03, median=1.04, upper quartile=1.10, upper whisker=1.19
        }, draw=black] coordinates {};

        \addplot[boxplot prepared={
            lower whisker=1.02, lower quartile=1.14, median=1.26, upper quartile=1.41, upper whisker=1.56
        }, draw=black] coordinates {};

        \addplot[boxplot prepared={
            lower whisker=1.13, lower quartile=1.19, median=1.27, upper quartile=1.30, upper whisker=1.41
        }, draw=black] coordinates {};

        \addplot[boxplot prepared={
            lower whisker=0.98, lower quartile=1.22, median=1.29, upper quartile=1.37, upper whisker=1.87
        }, draw=black] coordinates {};

        \addplot[boxplot prepared={
            lower whisker=1.16, lower quartile=1.19, median=1.21, upper quartile=1.24, upper whisker=1.26
        }, draw=black] coordinates {};

        \addplot[boxplot prepared={
            lower whisker=1.03, lower quartile=1.03, median=1.04, upper quartile=1.05, upper whisker=1.05
        }, draw=black] coordinates {};

        \end{axis}
    \end{tikzpicture}

    \caption{Chromium Speedup Distribution by Benchmark Suite.
    Geometric mean speedup across the 48 Blink benchmarks is $1.25\times$;
    suite-level medians range from $1.04\times$ (Image Decoder) to $1.29\times$ (Layout).
    The cumulative distribution of speedups is given in \protect\cref{fig:chromium-cdf}.}
    \Description{Box plots of the TSan+AllOpt speedup over Original TSan for each Chromium suite: Speedometer, Parser, DOM, SVG, Layout, Paint and Image Decoder. The median of every suite is above 1.}
    \label{fig:chromium_aggregated}
\end{figure}

\paragraph*{Impact of Thread Contention}
We measured overhead for Memcached, FFmpeg and SQLite with varying thread or CPU counts.
We found that the optimizations remain effective as the thread or CPU count increases.
For details, see~\cref{app:contention}.

\paragraph*{Instrumentation Reduction.}

\begin{figure*}
    \centering

    \pgfplotsset{
        allbenchstyle/.style={
            ybar,
            width=0.27\linewidth,
            height=3.5cm,
            bar width=0.35cm, %
            xtick=data,
            xticklabel style={
                font=\scriptsize, 
                anchor=east, 
                rotate=40,
                inner sep=1pt
            },
            ymin=0, ymax=100,
            tick style={draw=none},
            ylabel style={font=\scriptsize, yshift=-0.2cm},
            label style={font=\scriptsize},
            tick label style={font=\tiny}, %
            nodes near coords,
            nodes near coords style={
                font=\tiny,
                /pgf/number format/fixed,
                /pgf/number format/fixed zerofill,
                /pgf/number format/precision=1,
                anchor=south
            },
            clip=false,
            enlarge x limits=0.12
        }
    }

    \subfloat[Static Reduction (SIR)]{\label{fig:sir_reduction}
        \begin{tikzpicture}
            \begin{axis}[
                allbenchstyle,
                symbolic x coords={Memcached, Redis, FFmpeg, SQLite, MySQL, Chromium},
                ylabel={Sites Removed (\%)},
                ymin=0, ymax=16,
            ]
            \addplot[style={fill=red!56, draw=none}]
                coordinates {
                    (Memcached, 12.4)
                    (Redis, 7.1)
                    (FFmpeg, 2.7)
                    (SQLite, 1.2)
                    (MySQL, 0.6)
                    (Chromium, 0.8)
                };
            \end{axis}
        \end{tikzpicture}
    }\hfill
    \subfloat[Dynamic Reduction (DIR)]{\label{fig:dir_reduction}
        \begin{tikzpicture}
            \begin{axis}[
                allbenchstyle,
                symbolic x coords={Memcached, Redis, FFmpeg, SQLite, MySQL, Chromium},
                ylabel={Calls Avoided (\%)},
                ymin=0, ymax=120, bar shift=0pt,
            ]
            \addplot[style={fill=green!50!black!70, draw=none}, nodes near coords={}, error bars/.cd, y dir=both, y explicit]
                coordinates {(Memcached, 76.81) += (0,0.00) -= (0,0.00) (Redis, 91.23) += (0,0.07) -= (0,0.13) (FFmpeg, 65.95) += (0,0.02) -= (0,0.01) (SQLite, 5.31) += (0,2.02) -= (0,2.47) (MySQL, 3.58) += (0,0.09) -= (0,0.07) (Chromium, 65.02) += (0,0.04) -= (0,0.04)};
            \addplot[draw=none, fill=none, point meta=explicit symbolic, nodes near coords style={font=\tiny, anchor=south}]
                coordinates {(Memcached, 76.81) [76.8] (Redis, 91.30) [91.2] (FFmpeg, 65.97) [66.0] (SQLite, 7.33) [5.3] (MySQL, 3.67) [3.6] (Chromium, 65.06) [65.0]};
            \end{axis}
        \end{tikzpicture}
    }\hfill
    \subfloat[Dynamic Memory-Location\\Reduction (DMIR)]{\label{fig:dlir_reduction}
        \begin{tikzpicture}
            \begin{axis}[
                allbenchstyle,
                symbolic x coords={Memcached, Redis, FFmpeg, SQLite, MySQL, Chromium},
                ylabel={Locations Removed (\%)},
            , ymin=-10, ymax=110, bar shift=0pt]
            \addplot[style={fill=orange!56, draw=none}, nodes near coords={}, error bars/.cd, y dir=both, y explicit] 
                coordinates {(Memcached, 30.04) += (0,0.01) -= (0,0.01) (Redis, 27.93) += (0,1.07) -= (0,1.73) (FFmpeg, 86.92) += (0,0.02) -= (0,0.02) (SQLite, -2.70) += (0,3.36) -= (0,2.19) (MySQL, 0.28) += (0,0.50) -= (0,0.36) (Chromium, -0.63) += (0,0.68) -= (0,0.68)};
            \addplot[draw=none, fill=none, point meta=explicit symbolic, nodes near coords style={font=\tiny, anchor=south}]
                coordinates {(Memcached, 30.05) [30.0] (Redis, 29.00) [27.9] (FFmpeg, 86.94) [86.9] (MySQL, 0.78) [0.3]};
            \addplot[draw=none, fill=none, point meta=explicit symbolic, nodes near coords style={font=\tiny, anchor=south}]
                coordinates {(SQLite, 0.66) [$-2.7$] (Chromium, 0.05) [$-0.6$]};
            \end{axis}
        \end{tikzpicture}
    }\hfill
    \subfloat[Memory Ratio\\vs.~\tsan (MR)]{
      \label{fig:mem_overhead}
        \begin{tikzpicture}
            \begin{axis}[
                allbenchstyle,
                ymin=0.8, ymax=1.13, yticklabel style={/pgf/number format/fixed, /pgf/number format/fixed zerofill, /pgf/number format/precision=1},
                ytick distance=0.1, bar shift=0pt,
                symbolic x coords={Memcached, Redis, FFmpeg, SQLite, MySQL, Chromium},
                nodes near coords style={font=\tiny, anchor=south, /pgf/number format/fixed, /pgf/number format/fixed zerofill, /pgf/number format/precision=2},
            ]
            \addplot[style={fill=purple!42, draw=none}, nodes near coords={}, error bars/.cd, y dir=both, y explicit] coordinates {
                (Memcached, 0.928) += (0,0.021) -= (0,0.007) (Redis, 1.005) += (0,0.001) -= (0,0.000) (FFmpeg, 1.000) += (0,0.002) -= (0,0.002) (SQLite, 1.007) += (0,0.018) -= (0,0.015) (MySQL, 1.005) += (0,0.000) -= (0,0.000) (Chromium, 1.027) += (0,0.037) -= (0,0.041)
            };
            \draw[black, dashed, line width=0.5pt] (rel axis cs:0,0.60606) -- (rel axis cs:1,0.60606); %
            \addplot[draw=none, fill=none, point meta=explicit symbolic, nodes near coords style={yshift=-0.9pt}]
                coordinates {(Memcached, 0.949) [0.93] (Redis, 1.006) [1.01] (FFmpeg, 1.002) [1.00] (SQLite, 1.025) [1.01] (MySQL, 1.006) [1.01] (Chromium, 1.064) [1.03]};
            \end{axis}
        \end{tikzpicture}
    }

    \caption{Instrumentation Reduction for \texttt{TSan+AllOpt}. (a): sites removed by STC, SWMR, LO and EA.
    Higher is better; (d):~lower is better.}
    \Description{Four bar charts over Memcached, Redis, FFmpeg, SQLite, MySQL and Chromium: the percentage of instrumentation sites removed at compile time by STC, SWMR, LO and EA, and, for TSan+AllOpt, the percentage of instrumentation calls avoided at runtime, the percentage of instrumented memory locations removed at runtime, and the peak memory relative to Original TSan.}
    \label{fig:instr_reduction}
\end{figure*}

Runtime performance gains result from our optimizations pruning instrumentation. \cref{fig:instr_reduction} quantifies this reduction from three perspectives: the reduction of instrumentation sites at compile time (SIR), of instrumentation calls at runtime (DIR), and in unique instrumented memory locations (DMIR).
Our framework statically identifies and removes unnecessary instrumentation (\cref{fig:sir_reduction}): STC, SWMR, LO and EA remove up to 12.4\% of instrumentation sites (Memcached). This demonstrates that our compile-time analyses prove many accesses race-free at no runtime cost.
At runtime, the decrease is larger: for Memcached and Redis, 77\% and 91\% fewer instrumentation calls happen (\cref{fig:dir_reduction}). The runtime figure also counts calls avoided by DE and by the dynamic tests of \cref{sec:trans}, which skip calls at sites that remain instrumented.

Our optimizations reduce the number of unique memory locations instrumented 
(\cref{fig:dlir_reduction}). For FFmpeg and Redis, we remove instrumentation for 86.9\% and 27.9\% of locations, respectively.

\paragraph*{Memory Consumption.}
The peak memory of \code{TSan+}\allowbreak\code{AllOpt} stays within 3\% of Original \tsan{}'s (\cref{fig:mem_overhead}), except for Memcached (7\% lower).
Instrumenting fewer accesses can lower memory use, as \tsan lazily allocates part of its shadow memory on the first checked access.
In FFmpeg, nearly all locations no longer checked (\cref{fig:dlir_reduction}) belong to its single-threaded stream copy, whose peak memory falls by 11\%; \cref{fig:mem_overhead} shows FFmpeg's largest process, an H.265 encode, whose peak is unchanged.

\newpage
\subsection{Compilation Overheads and Race Detection Capabilities (\RQ{2})}

\begin{figure}
    \centering

    \pgfplotsset{
        allbenchstyle/.style={
            ybar,
            width=0.9\linewidth,
            height=3.15cm,
            bar width=0.38cm,
            xtick=data,
            xticklabel style={
                font=\tiny,
                anchor=north,
                rotate=0,
                inner sep=0.5pt
            },
            ymin=0, ymax=100, yticklabel={\pgfmathprintnumber{\tick}\%}, %
            tick style={draw=none},
            ylabel near ticks,
            ylabel style={font=\small},
            label style={font=\small},
            tick label style={font=\scriptsize},
            nodes near coords,
            nodes near coords style={
                font=\scriptsize,
                /pgf/number format/fixed,
                /pgf/number format/precision=3,
                anchor=south
            },
            clip=false,
            enlarge x limits=0.08
        }
    }

        \begin{tikzpicture}
            \begin{axis}[
                allbenchstyle,
                ymode=log, ymin=1, ymax=600,
                log ticks with fixed point,
                yticklabel=\axisdefaultticklabellog,
                ylabel={\scriptsize Added build time (s)},
                point meta=rawy,
                nodes near coords={\pgfmathprintnumber[fixed,precision=0]{\pgfplotspointmeta}\,s},
                symbolic x coords={Memcached, Redis, FFmpeg, SQLite, MySQL},
            ]
            \addplot[style={fill=teal!70, draw=none}] coordinates {
                (Memcached, 2.0) (Redis, 12.3) (FFmpeg, 10.2) (SQLite, 6.8) (MySQL, 131.0)
            };
            \end{axis}
        \end{tikzpicture}

    \caption{
      Build time added by \texttt{TSan+AllOpt} over a \tsan build (wall time, median of 3 builds, 8 jobs; MySQL 6) (lower is better).
    }
    \Description{Bar chart with a logarithmic vertical axis of the build time, in seconds, that TSan+AllOpt adds to a TSan build of Memcached, Redis, FFmpeg, SQLite and MySQL. Memcached has the smallest addition and MySQL the largest.}
    \label{fig:overheads}
\end{figure}

\paragraph*{Compilation Time.}
\code{TSan+AllOpt} adds between a few seconds and 2.2 minutes to a build (\cref{fig:overheads}); on Memcached, one run under \code{TSan+AllOpt} is 111 seconds shorter, so the extra build time is recovered in the first run.
At the measured speedups, it is recovered after about a minute of testing on Redis and SQLite, and after about half an hour on MySQL.

\paragraph*{Race Detection}

On Memcached, SQLite, MySQL and Chromium, where Original \tsan reports races, we ran it and \code{TSan+AllOpt} at least ten times each and matched races by the pair of functions involved, which also covers reports that DE relocates to a (post-)dominating access in the same function (\cref{sec:beyond}).
On the first three, every race that Original \tsan reported in at least four runs was also reported by \code{TSan+AllOpt}.
On Chromium, \code{TSan+AllOpt} reported each of the three races that Original \tsan reported.
On Redis and FFmpeg, Original \tsan reports no race under our workloads.
We also ran \tsan's regression suite (about 400 tests) with the flags of each configuration: no test failed because of a missed race.
This confirms that our optimizations preserve TSan's race detection.

\subsection{Studying Static Analysis Limitations (\RQ{3})}\label{ss:limitations}

\paragraph{Oracle Tracer.}
To estimate how many \emph{dynamic} accesses need a check, we develop a dynamic tracer identifying shared-memory events executed while the program is multi-threaded (see \cref{app:oracle-tracer}).
Inspecting the results revealed sources of analysis imprecision:
complex data flow, globally reachable heap, function pointers, runtime-dependent addressing, and inline assembly
(see \cref{tab:imprecision_causes}), suggesting directions for future improvements.

\paragraph{Comparison with a Dynamic Redundancy Filter (\textsc{ReX})}
To contrast our mostly static approach with a dynamic one, we implemented the %
\textsc{ReX} filter~\cite{huang2017rex} for \tsan.
\textsc{ReX} computes runtime signatures of each access's \emph{concurrency context} and filters redundant accesses.
It proved highly effective, confirming that most dynamic memory access events are redundant: over 95\% were filtered in Memcached and SQLite, and around 37\% in Redis.
However, this did not translate into a net performance gain in our implementation: the filter made benchmarks several times slower than even the original \tsan (see \cref{app:rex-details}).
This highlights a key advantage of our approach: it avoids the high per-access runtime cost of dynamic filtering,
but at the expense of being more conservative.%

\begin{tcolorbox}[mygraybox]
\looseness=-1
Different optimizations are strong on different programs, which makes their combination a good suite to have in \tsan.
Other analyses eliminate instrumentation where DE is inapplicable, such as for non-global objects or read-only global variables (in Redis and Memcached).
\end{tcolorbox}

\section{Related Work}\label{sec:related}
\emph{Dynamic} and \emph{static} race detection offer complementary trade-offs:
dynamic analyses are precise but incur runtime overhead, whereas sound static analyses can prove absence of races but report many false positives. These can be reduced by optimistic assumptions, but only at the cost of soundness.
Our approach combines the best of both worlds:
we integrate lightweight, sound static analyses into the commercial-grade \tsan implementation in LLVM, reducing instrumentation and runtime overhead while preserving its guarantees and workflow.

\paragraph*{Dynamic Techniques.}
Happens-before race detectors such as \tsan~\cite{serebryany2009threadsanitizer} and \textsc{FastTrack}~\cite{flanagan2009fasttrack,wilcox2018verifiedft} are precise but costly; \textsc{FastTrack}'s epoch-based representation is also used by \tsan~\cite{TSanAlgorithm}.
Efficient data structures such as tree clocks~\cite{treeClocksMathur2022} and CSSTs~\cite{tunc2024cssts} reduce the cost of maintaining partial orders.
Recent work studies dynamic and predictive analyses, which infer races~\cite{kini2017dynamic,mathur2018shb,roemer2018high,roemer2020smarttrack,cai2021sound,pavlogiannis2019fast,mathur2020complexity,mathur2021optimal,kulkarni2021finegrained,shi2024optimistic,ang2024strong} or other violations~\cite{tunc2023deadlock,mathur2020atomicity,tunc2026atomicity,ang2024pattern} from observed traces.
Complementary concurrency-testing approaches improve schedule exploration through reads-from-guided fuzzing~\cite{wolff2024greybox} or selectively uniform sampling of interleavings~\cite{zhao2025selectively}.
Our work is complementary to these directions.
Hardware tracing, such as Intel PT, can also reduce instrumentation overhead~\cite{tmon}; our static elimination is compatible with such mechanisms, since eliminated accesses need not be collected.
\emph{Sampling} reduces overhead by instrumenting only a subset of accesses, at a fixed rate~\cite{bond2010pacer}, adaptively per thread~\cite{marino2009literace}, with detection algorithms designed for samples~\cite{UTrack2025,RPT2023}, or with hardware support~\cite{erickson2010effective,ShengVEHCZ11,cai2016deployable,ZhangJL17},
trading completeness for performance.
Our work instead reduces the overhead while preserving completeness --
hopefully making sampling obsolete for many programs.
Chabbi and Ramanathan~\cite{DBLP:conf/pldi/ChabbiR22} show that at industrial
scale (Uber's Go codebase), TSan's overhead makes CI-time detection
infeasible and explicitly call for CI-ready algorithms.
Our overhead reductions are a direct step in that direction.

\paragraph*{Static Techniques.}
We focus on the most closely related dataflow-based static race detectors, while noting approaches based on model checking~\cite{he2021satisfiability,gavrilenko2019bmc,beyer2011cpachecker,toth2017theta} and automata-based verification~\cite{heizmann2013software,farzan2022sound,dietsch2023ultimate}.
Static race detectors span cheap compositional analyses~\cite{engler2003racerx,blackshear2018racerd} and more precise but heavyweight, sometimes context- or flow-sensitive techniques~\cite{naik2006effective,pratikakis2011locksmith,kahlon2007fast,liu2021o2,vojdani2016static,schwarz2025digest,SchwarzE24,CuiGUPLH23,MaSH23,SwainLLLG020,LiL019,aplas26}.
The former scale better but report many false positives, sometimes mitigated by unsound post-processing; the latter are more precise but often do not scale.
A recent survey~\cite{holter2025sound} highlights that common idioms
such as thread pools are still out of reach for static analyzers.
Moreover, many are not integrated into compiler pipelines, requiring non-trivial setup.

\paragraph*{Hybrid Techniques.}
Prior hybrid detectors also reduce instrumentation with static analysis.
A central distinction of our approach is its integration into existing LLVM build workflows: it needs no whole-program pointer analysis, and scales to Chromium.
Early hybrid detectors~\cite{choi2002efficient,vonPraun2003static} use whole-program
analyses to identify accesses that cannot race and omit their
instrumentation. These systems, as well as \textsc{RedCard} and \textsc{BigFoot}
discussed below, target a managed runtime and can exploit runtime-system support
unavailable in our setting.
Di et al.~\cite{di2016accelerating}, like us, target LLVM and \tsan for C/C++, computing
may-happen-in-parallel instruction pairs and filtering them to potentially racy
pairs.
They rely on whole-program (LTO) pointer analysis as provided by \textsc{SVF}, which previous work
reports~\cite{kuderski2019unification,yao2021efficient} takes hours on LLVM-sized programs and does not finish within 12 hours on MySQL or Firefox.
Our analyses operate per translation unit without requiring LTO; making compile-time overheads
much more manageable.\footnote{We could not compare empirically with their system because its source code
is unavailable, and the authors did not respond to our inquiries.}
Das et al.~\cite{das2015section} use section-based program analysis to reduce memory-access instrumentation for tools including \tsan; our work targets a different design point, using lightweight compiler analyses integrated into LLVM \tsan.
Early \emph{dominance elimination} removes intraprocedurally dominated accesses with no intervening function calls~\cite{choi2002efficient}.
\textsc{RedCard}~\cite{flanagan2013redcard} generalizes this to checks on locations already accessed within the same code fragment without intervening release-like operations.
\textsc{Bigfoot}~\cite{rhodes2017bigfoot} further generalizes redundant checks, akin to our post-dominance, and aggressively moves checks to expose redundancy, e.g., from inside a loop to after it.
This leads to reporting locations where no access occurs, harming usability; our reported locations always correspond to racy access instructions.
Both tools target Java. %

\textsc{RedCard} and \textsc{BigFoot} also compress shadow memory. Metadata compression is not directly applicable to \tsan, whose checks must respect aligned 8-byte granules. Nevertheless, eliminating instrumentation can reduce shadow memory indirectly.

Our approach unifies both removing instrumentation for provably safe accesses and eliminating
redundant instrumentation, and combines them with a novel lightweight way to skip instrumentation when the
access is executed in a single-threaded context.
It is being integrated into \tsan
which ships with the \textsc{LLVM} compiler infrastructure. It supports compositionality at the
level of compilation units.
Our redundancy elimination
strikes a good balance between aggressively eliminating redundant checks and keeping the reports
immediately actionable to developers.

\textsc{RaceMob}~\cite{racemob2013} crowdsources instrumentation: static analysis identifies potentially racy accesses, and each client instruments a subset and reports detected races to a server. This deployment strategy is complementary to ours.
\textsc{IFRit}~\cite{Effinger-DeanLCGB12} statically identifies interference-free regions and monitors which regions are active. A race is reported when conflicting accesses have overlapping regions. Since overlapping regions are sufficient but not necessary for a race, \textsc{IFRit} can miss races.
Other hybrid approaches target structured concurrency where high-level insights enable cheap analyses~\cite{atzeni2014lowoverhead,archerLong,schwitanski2023leveraging}, binaries~\cite{sun2024hardracedynamicdatarace}, or driver-specific locking-discipline violations~\cite{DBLP:journals/tse/BaiCJLH22}. For OpenMP, Jammer et al.~\cite{jammer2026justraces} drop computation that race detection does not need while keeping \tsan's detection accuracy; the program then no longer computes its results, and their whole-program analysis roughly triples build time. Such settings differ from our goal of improving general-purpose \tsan instrumentation in LLVM.
\textsc{SDILP}~\cite{DBLP:journals/tse/BaiCJLH22} additionally uses dynamic alarms to guide further static race reports.

Unlike approaches that use optimistic, profiling-derived invariants and recover soundness by monitoring violations and re-executing with conservative instrumentation~\cite{devecsery2018optimistic}, we rely on sound static analyses, and on Memcached and SQLite on runtime checks, without profiling or re-execution. 
Such profiling and re-execution require workflow changes often unacceptable in
production settings.

\section{Conclusion}
\label{sec:conclusion}

To address the high overhead of dynamic race detectors, we developed a
lightweight framework that integrates multiple static analyses into the
compilation pipeline, reducing instrumentation without compromising soundness or
completeness.
We instantiated this framework with five complementary static analyses--escape
analysis, lock ownership, single-threaded context, SWMR detection, and
dominance-based redundancy elimination--and a transformation that dynamically
skips instrumentation in single-threaded contexts.
Our evaluation on large, real-world applications confirms that this approach significantly reduces \tsan's overhead, achieving application-level speedups of up to $1.91\times$, a geometric-mean Chromium speedup of $1.25\times$ on Blink and $1.14\times$ on Speedometer~3.0, and peaks of up to $2.28\times$ on Memcached with two server threads. 
On four applications, the extra build time is recovered after minutes of testing under \tsan.
This underscores the
practicality of our approach.
Our optimizations are in the process of being upstreamed into \tsan.

\paragraph{Acknowledgments}
This work is partially supported by the National Research Foundation, Singapore, and Cyber Security Agency of Singapore under its National Cybersecurity R\&D Programme (Fuzz Testing <NRF-NCR25-Fuzz-0001>). Any opinions, findings and conclusions, or recommendations expressed in this material are those of the author(s) and do not reflect the views of National Research Foundation, Singapore, and Cyber Security Agency of Singapore.

\newpage
\bibliographystyle{ACM-Reference-Format} 
\bibliography{references} 
\label{paper-end}

\clearpage
\setcounter{figure}{0}\renewcommand{\thefigure}{S\arabic{figure}}\renewcommand{\theHfigure}{S\arabic{figure}}
\setcounter{table}{0}\renewcommand{\thetable}{S\arabic{table}}\renewcommand{\theHtable}{S\arabic{table}}
\appendix
\section{Concrete Semantics}
\cref{fig:definitions-notation,fig:instr-set} summarize our notation and the instruction set of our fragment of a {C}-like language, which we use for our formal treatment of the analyses and their soundness proofs.

For a program $P$, we denote its concrete semantics by $\sem{P}$: the set of all traces described in \cref{ss:concrete-semantics}.
This set includes both complete traces and traces corresponding to partial executions.

\begin{figure*}
$
\begin{array}{@{}ll@{\qquad}ll@{}}
    \Locks & \text{Set of locks}
    & \BBs_f & \text{Basic blocks in function } f \in \Functions \\

    \Globals & \text{Global variables}
    & \Formals_f & \text{Formal parameters of } f,\ \Formals_f \subseteq \Registers_f \\

    \Functions & \text{Functions}
    & 
    \func\,i &\text{Function containing instruction } i 
   \\

    \Instructions & \text{Instructions (see \cref{fig:instr-set})}
    &
     \bblock\,i & \text{Basic block containing instruction } i 
    \\

    \Locals_f & \text{Local variables of } f \in \Functions
    & \memAccessInstructions & \text{All memory access instructions in a program} \\

    \Registers_f & \text{Virtual registers of } f \text{ (do not escape)}
    & \locs{i} & \text{Memory locations accessed by instruction } i
\end{array}
$
\setlength{\arraycolsep}{5pt}
\caption{Definitions and notation for our fragment of a {C}-like language. All sets are finite.}
\Description{Two-column table of the notation of the formal model: the sets of locks, global variables, functions, instructions and memory access instructions; for each function, its local variables, virtual registers, basic blocks and formal parameters; and the maps from an instruction to its function, its basic block and the memory locations it accesses.}
\label{fig:definitions-notation}
\end{figure*}
\begin{figure*}
\setlength{\arraycolsep}{2pt}
$
\begin{array}{ll}
    \lock{l}\;/\;\unlock{l} & \text{Locking/unlocking a lock } l \in \Locks\\
    a_0 \texttt{ = f(}a_1,\dots,a_{n}\texttt{)}; & \text{Static call to an $n$-ary
        function $f\in\Functions$ with arguments $a_1,\dots,a_{n}$}\\
        & \text{from $\Registers_g$ where $g$ is the calling function, assigning result
        to $a_0 \in \Registers_g$}\\
    a_0 \texttt{ = *$a_{1}$(}a_2,\dots,a_{n+2}\texttt{)};\,\,\,\,\,\, & \text{Dynamic call 
        with function pointer stored in $a_{1} \in \Registers_g$}\\
    \texttt{return }a; & \text{Return from current function, returning value in register } a \in \Registers_f\\
    a_0 = \mathbf{F}(a_1, \cdots, a_{i}); & \text{Computations involving only registers}\\
    \create{f}; & \text{Create a new thread executing an argumentless function } f \in \Functions\\
    a_0 = r(*a_1); & \text{Read value from memory location pointed to by } a_1 \in \Registers_f 
    \text{ and store it into } a_0 \in \Registers_f\\
    w(*a_0, a_1); & \text{Write value of register } a_1 \in \Registers_f \text{ into memory location pointed to by } a_0 \in \Registers_f\\
    a = \&v; & \text{Copy address of global or local variable } v \in (\Globals \cup \Locals_f) \text{ into } a \in \Registers_f\\
    a = \&g; & \text{Copy address of function } g \in \Functions \text{ into } a \in \Registers_f
\end{array}
$
\setlength{\arraycolsep}{5pt}
\caption{Instruction set $\Instructions$ of our fragment of a {C}-like language.}
\Description{Table of the instructions of the formal language and their meaning: lock and unlock, static and dynamic calls, return, computation on registers, thread creation, memory read and write through a register, and taking the address of a variable or a function.}
\label{fig:instr-set}
\end{figure*}

\section{Sufficient Instrumentation, Formally}
As our proofs provided in this appendix require some additional rigor which we have omitted in the main text for readability,
we first provide refined definitions for some concepts from~\cref{sec:sufficient}.

\begin{definition}[Instrumented Memory Accesses]
    Let $P$ be a program and $\memAccessInstructions$ the set of all memory access instructions in $P$.
    We assume that the analyzer takes as an additional parameter a set of memory accesses 
    $\memInstr \subseteq \memAccessInstructions$, representing the instructions to be instrumented.\footnote{Such sets depend on the program $P$. We do not make this dependence explicit
        in the notation in the interest of readability. 
    }
\end{definition}
\begin{definition}[Completeness of a Dynamic Data Race Detector w.r.t.\ Instrumented Accesses]
    We call a dynamic data race detector \emph{complete w.r.t.\ a set of
    instrumented accesses $\memInstr$}, iff, for any program $P$ and any of its traces $t\in \sem{P}$ the following holds: 
    For any race involving dynamic events corresponding to instructions $i_x$ and $i_y$ in $t$, the analyzer will, 
    when analyzing $t$, flag the race when $\{i_x,i_y\} \subseteq \memInstr$.
\end{definition}
\begin{definition}[Sufficient Instrumentation Set]
    For a program $P$ and its set of memory access instructions $\memAccessInstructions$, 
    the set of minimal required memory instrumentations $\memInstrSuff$ is given by
    \[
        \memInstrSuff = \text{$\bigcup$} \{ \{x,y\} \mid x,y \in \memAccessInstructions, 
            \exists t \in \sem{P}: x \text{ and } y \text{ race in } t   \}.
    \]
    $\memInstrSuff$ corresponds to those memory accesses that are actually involved in data races in at least one of the executions of the program.
    We call any set $\memInstr$ with $\memInstrSuff \subseteq \memInstr$ a \emph{sufficient instrumentation set} for $P$.
\end{definition}

\section{Details and Proofs for the Analysis}\label{app:details}
Here, we provide a more formal treatment of the details of our analyses.
\subsection{Details on Single-Threaded Context (STC) Identification}\label{app:stc}
We first recall \cref{prop:stc-conc} before providing its proof.
\begin{proposition}
    A memory access instruction $x$ that is guaranteed to only be executed before the program creates additional threads cannot be involved in any data race, and thus $x \notin \memInstrSuff$.
\end{proposition}
\begin{proof}
    Consider for a contradiction a trace $t$ with a dynamic event $e_x$ corresponding to $x$ and a dynamic event $e_y$ such that $e_x$ and $e_y$ race.
    If $e_x$ and $e_y$ belong to the same thread, they are ordered by $\to_{po}$ (contradiction).
    If they belong to different threads, $e_y$ occurs in a thread (transitively) created by the main thread
    after $e_x$ occurred. $e_x$ and $e_y$ are thus ordered by $(\to_{po} \cup \to_{cr})^*$ (contradiction).  
\end{proof}
The set of functions and \mainFun's basic blocks that may potentially execute
in a multi-threaded context ($\Functions^{MT}$ and $\BBs_{\mainFun}^{MT}$, respectively), are then described by the least fixed-point
solution of the following equations:
\small
\begin{equation*}
\begin{array}{lll}
    \Functions^{MT}_{c} &\supseteq& 
            \{ f \in \Functions \setminus \{\mainFun\} \mid
                f \text{ contains } \texttt{create} \lor\\
            && \qquad\qquad\qquad
                (\exists g \in \Functions^{MT}_{c}:
                f \text{ may call } g)
            \}\\[2ex]

        \Functions^{MT} &\supseteq& 
            \Functions^{MT}_{c}\\
            && \cup   \{ f \in \Functions \setminus \{\mainFun\} \mid
                f \text{ has its address taken}
            \}\\
            && \cup   \{ f \in \Functions \setminus \{\mainFun\} \mid
                \exists g \in \Functions^{MT}:
                g \text{ may call } f
            \}\\
            && \cup    \{ f \in \Functions \setminus \{\mainFun\} \mid
                \exists b \in \BBs_{\mainFun}^{MT}:
                b \text{ may call } f
            \}
            \\[2ex]
            
        \BBs_{\mainFun}^{MT} &\supseteq& 
            \{ b \in \BBs_{\mainFun} \mid
                b \text{ contains } \texttt{create}
            \}\\
        &&  \cup \{ b \in \BBs_{\mainFun} \mid
                \exists f \in \Functions^{MT}_{c}:
                b \text{ may call } f
            \} \\
        &&  \cup \{ b \in \BBs_{\mainFun} \mid
                \exists b' \in \BBs_{\mainFun}^{MT}:\\
        && \qquad\qquad\qquad\qquad  b \text{ is a CFG successor of } b'
            \}
    \end{array}
\end{equation*}
\normalsize
where the \emph{may call} relationship is an overapproximation of the actual behavior of the program, as computed by the analyzer in a first step.
Here, $\Functions^{MT}_{c}$ is an auxiliary set collecting functions that create threads or call other functions that do.
Calling a function in $\Functions^{MT}_{c}$ means that the caller needs to be marked as multi-threaded,
while this is not the case for a function that sometimes is called by a multi-threaded function but does not create threads itself, nor call functions that do.
\begin{proposition}\label{th:stc}
    Provided the \emph{may call} relationship is a sound overapproximation of the
    program's call behavior, any memory access instruction $x$ occurring in a function
    $f$ or in a basic block $b$ with $f \notin \Functions^{MT}$ and $b \notin \BBs_{\mainFun}^{MT}$
    does not occur after any thread creation.
\end{proposition}
\begin{proofsketch}
    As this is the first proof sketch, we provide some more details here, while we will omit these later.
    Consider a constraint system ranging over non-main functions and basic blocks of \mainFun,
    where for each $x \in \Functions \setminus \{\mainFun\} \cup \BBs_{\mainFun}$,
    $[x]$ collects the set of traces where at least one
    thread is in the corresponding function or basic block.
    This constraint system can be shown equivalent to the trace semantics of the program.
    We then construct from $\Functions^{MT}$ and $\BBs_{\mainFun}^{MT}$ a mapping
    \[
        \eta\,[x] =
        \left\{
        \begin{array}{@{}l@{}}
            \begin{aligned}
                \{\, t \in \sem{P} \mid{}&
                    \texttt{create(\_)}\not\in t,\\
                    & \text{at least one thread at } x \,\}
            \end{aligned}
            \\[1ex]
            \quad\text{if } x \not\in{}(\Functions^{MT} \cup \BBs_{\mainFun}^{MT})
            \\[2ex]
            \{\, t \in \sem{P} \mid \text{at least one thread at } x \,\}
            \\[1ex]
            \quad\text{otherwise}
        \end{array}
        \right.
    \]
    and show that $\eta$ is a solution to the above constraint system.
\end{proofsketch}

\subsection{Details on the SWMR Analysis}\label{app:swmr-proof}
\begin{proposition}
    A memory read instruction $x$ accessing  memory locations $L$, cannot be involved in any race
    if all instructions $y$ potentially writing to any location in $L$ can only occur before 
    the program creates additional threads, i.e., all such $y$ occur in ST context. Thus,
    $x \notin \memInstrSuff$.
\end{proposition}
\begin{proofsketch}
    Consider an instruction $x$ potentially reading memory locations $L$, where all instructions
    potentially writing to any location in $L$ occur before the program creates additional threads.
    Consider for a contradiction a trace $t$ in which a dynamic event $e_x$ corresponding to a read access at $x$
    to some memory location $\ell \in L$ and another dynamic event $e_y$ accessing the same memory location $\ell$
    race.
    Then, $e_y$ must correspond to a write and the corresponding access instruction $y$ writes to $\ell\in L$.
    If $e_x$ and $e_y$ are on the same thread, they are ordered by $\to_{po}$ (contradiction).
    If they are on different threads, $e_y$ occurs on the main thread, and $e_x$ occurs in a thread (transitively) created by the main thread after $e_y$ occurred. $e_y$ and $e_x$ are thus ordered by $(\to_{po} \cup \to_{cr})^*$ (contradiction).  
\end{proofsketch}
Let us assume the sound helper analysis computes
an overapproximation of the set of access instructions $\writes{\ell}$ for each memory location $\ell$,
as well as an overapproximation of the set of memory locations $\locs{i}$ accessed by each access instruction $i$.\footnote{
    In practice, sets are from an abstract domain, applying, e.g.,
    allocation-site abstraction to dynamic allocation.
}
\begin{proposition} 
Then, a memory access instruction $x$ need not be instrumented if:
\[
    \begin{array}{l}
    \forall \ell\in \locs{x}: \forall y \in \writes{\ell}: \\
    \qquad
        (\func\,y) \not\in \Functions^{MT} \land (\bblock\,y) \not\in \BBs_{\mainFun}^{MT}, 
    \end{array} 
\]
i.e., all writes to locations potentially accessed by $x$ occur in ST context.
\end{proposition}
\begin{proofsketch}
    By \cref{prop:swmr}, soundness of the pointer analysis, and \cref{th:stc}.
\end{proofsketch}

\subsection{Details on the LO Analysis}\label{app:lo-proof}
\begin{proposition}
    A memory access instruction $x$ happening while holding the set of locks $S_0$ that potentially accesses memory locations $L$, cannot be involved in any race if all instructions $y$ potentially accessing any location in $L$
    either happen in ST context or happen while holding a set of locks $S_1$ with $S_0 \cap S_1 \neq \emptyset$,
    i.e., they share at least one common lock.
\end{proposition}
\begin{proofsketch}
    Consider for a contradiction a trace in which a dynamic event $e_x$ corresponding to an access instruction $x$ is involved in a race on some memory location
    $\ell$.
    Then, there is another dynamic event $e_y$ accessing the same memory location $\ell$ which corresponds to an access instruction $y$.
    If $y$ or $x$ happens in ST context, we have a contradiction (by \cref{th:stc}). Assume both happen in an MT context. $e_x$ happens while holding the set of locks $S_0$, and $e_y$ happens while holding the set of locks $S_1$ with $S_0 \cap S_1 \neq \emptyset$. Choose any $l \in S_0 \cap S_1$.
    Then, the thread executing $e_x$ must have acquired $l$ before executing $e_x$, and the thread executing $e_y$ must have acquired $l$ before executing $e_y$.
    As only one thread can hold a lock at any given time, either the thread executing $e_x$ released $l$ before $e_y$ acquired it, or vice versa.
    Thus, the two events are ordered by the transitive closure of $\to_{po}$ and $\to_{lo}$ (contradiction).
\end{proofsketch}
\begin{corollary}
    A memory access instruction $x$ need not be instrumented if, for all memory locations $\ell$ potentially accessed by $x$, there is a non-empty set of locks $S(\ell)$ such that for any access instruction $y$ potentially accessing $\ell$ happening in a multi-thread context, $S(\ell)$ is held.
\end{corollary}
This set $S(\ell)$ is said to \emph{own} the memory location $\ell$.
\begin{proposition}
A memory access instruction $x$ need not be instrumented if $\;\forall \ell\in \locs{x}: S(\ell) \neq \emptyset$,
i.e., all accessed locations are owned by one of the held locks.
\end{proposition}
\begin{proofsketch}
    By \cref{prop:protection}, soundness of pointer and lockset analyses, and \cref{th:stc}.
\end{proofsketch}

\subsection{Details on Escape Analysis}\label{app:ea}
\begin{proposition}
    A memory access instruction $x$ accessing only memory locations that are local to its thread cannot
    be involved in any race. Thus, $x\not\in\memInstrSuff$. 
\end{proposition}
\begin{proofsketch}
    Consider for a contradiction a trace in which a dynamic event $e_x$ corresponding to an access instruction $x$ is involved in a race on some memory location $\ell$.
    Then, there is another dynamic event $e_y$ accessing $\ell$ which corresponds to an access instruction $y$.
    As $\ell$ is local to the thread executing $e_x$, $e_y$ must be executed by the same thread as $e_x$.
    Thus, $e_x$ and $e_y$ are ordered by $\to_{po}$ (contradiction).
\end{proofsketch}
We next relate this to the high-level rules (1)-(4) given in \cref{ss:ea} about when an object is considered to have escaped.
\begin{proposition}
    An object $o$ that has not escaped according to rules (1)-(4) given in \cref{ss:ea} remains thread-local,
    and thus any access to it is performed by the thread that created it.
\end{proposition}
\begin{proofsketch}
    For a thread to access $o$, it must obtain a pointer to it. We verify that if none of the conditions above are met, no pointer to $o$ becomes accessible to other threads.
\end{proofsketch}

\paragraph{The use-after-free premise}\label{app:uaf}
The argument covers pointers derived from $o$, not a dangling pointer to the memory $o$ occupies.
If a thread keeps a pointer to a freed heap block, the allocator then happens to reuse the block for a provably non-escaping object $o$, and the thread then writes through the stale pointer while the owner writes $o$, stock \tsan reports a race for this use-after-free, whereas our optimized version does not. 
The same holds for a stack slot reused after its variable's lifetime ends.

The analysis operates in two passes over the call graph (illustrated in \cref{fig:ipa_diagram}):
\begin{itemize}[leftmargin=*]
    \item \textbf{Bottom-up Pass (blue):} The analysis begins at the leaves of the call graph (Steps \textcircled{1} and \textcircled{2}). For \code{func_B} and \code{func_C}, it generates summaries indicating what happens to their arguments. For instance, the summary for \code{func_B} states that parameter \code{q} escapes, while the one for \code{func_C} states that \code{r} does not escape. At Step \textcircled{3}, the caller \code{func_A} uses these summaries to determine its own summary.
    \item \textbf{Top-down Pass (green):} At Step \textcircled{4}, the analysis uses context from the caller. Since \code{x} is passed to an escaping function in \code{main}, it is marked as escaping. This context is then propagated downwards. At Step \textcircled{5}, when analyzing the call to \code{func_A}, \code{y} is also marked as escaping.
\end{itemize}
These phases are iterated until a fixed point is reached. This bidirectional flow, combined with
field-sensitivity, allows for highly precise identification of non-escaping data.

While we have given the high-level idea for our escape analysis and an illustrative example in
\cref{ss:ea}, we provide a more formal description here. We remark that while our implementation
and the informal description in the main text are field-sensitive, we omit this feature here for
simplicity of presentation.
Our Escape Analysis (EA) here does not rely on the results of a points-to-analysis and computing a global set of escaped memory locations.
Instead, for each auxiliary $a$ it keeps track of whether the auxiliary can contain the address of a global or a variable that has escaped.
Any addresses that are read from a global or from an auxiliary marked as escaped are also treated as potentially escaped.
\newcommand{\DangerousRegisters}{\Registers^{esc}}
\newcommand{\EscapedLocals}{\Locals^{esc}}
Let $\Instructions_f$ for a function $f\in \Functions$ be the set of instructions appearing in $f$.
We then compute for each function $f$ two sets
\begin{itemize}
    \item $\DangerousRegisters_f \subseteq \Registers_f$: the set of registers that may contain pointers to escaped memory locations (either globals or escaped locals).
    \item $\EscapedLocals_f \subseteq \Locals_f$: the set of local variables of $f$ that may have escaped. 
\end{itemize}
and a summary $\Sigma_f$ which is a tuple of booleans indicating, for each formal argument of $f$, whether the function may cause the argument to escape when passed to $f$.
The resulting constraint system is then given in \cref{fig:ea-constraints}.
\begin{figure*}
\[
    \begin{array}{lllr}
        \DangerousRegisters_f &\supseteq&  \{ a \mid (a = \&v) \in \Instructions_f, v \in (\Globals \cup \EscapedLocals_f ) \}  \cup\; \{ a_0 \mid (a_0 = r(*a_1)) \in \Instructions_f, a_1 \in \DangerousRegisters_f \} \\
        \quad && \cup\; \{ a_1 \mid (w(*a_0,a_1)) \in \Instructions_f, a_0 \in \DangerousRegisters_f \} \\ 
        \quad && \cup\; \{ a_0 \mid (a_0 = \mathbf{F}(a_1,\dots,a_i)) \in \Instructions_f, \{a_1, \dots, a_i\} \cap \DangerousRegisters_f \neq \emptyset \} \cup\; \{ a \mid (\text{return }a)\in \Instructions_f \}\\ 
        \quad && \cup\; \{ a_0 \mid (a_0 \texttt{ = g(}a_1,\dots,a_{n}\texttt{)})\in \Instructions_f  \}\\
        \quad && \cup\; \{ a_i \mid (a_0 \texttt{ = g(}a_1,\dots,a_{n}\texttt{)})\in \Instructions_f, i \in [1,n], \langle \Sigma_g \rangle_i \} & (*)\\
        \quad && \cup\; \{ a_{i} \mid (a_0 \texttt{ = f(}b_1,\dots,b_{n}\texttt{)})\in \Instructions_g, a_{i} \text{ is $i$-th formal of f}, b_{i+1} \in
            \DangerousRegisters_g  \} & (\ddagger)\\[2ex]

        \EscapedLocals_f \supseteq \{ v \mid (a = \&v) \in \Instructions_f, a \in \DangerousRegisters_f, v \in \Locals_f \}\span\span \\[2ex]

        \Sigma_f = (a_0 \in \DangerousRegisters_f, \dots, a_{n-1} \in \DangerousRegisters_f)
        \qquad \text{(for $a_i$ the $i$-th formal of function $f$)}
        \span\span  & (\dagger)
    \end{array}
\]
\caption{Constraint system for escape analysis where $\Sigma_f$ for an
$n$-ary function $f$ is a tuple of booleans indicating whether the function $f$ will
cause its $k$-th argument to escape, and $\langle \cdot \rangle_i$ denotes accessing the $i$-th element of the tuple.}\label{fig:ea-constraints}
\Description{System of set constraints that defines, for each function, the registers that may point to escaped memory, the local variables that may escape, and the summary of which arguments the function may let escape. The rules propagate this information from callees to call sites and from call sites to callees.}
\end{figure*}
For clarity, we have omitted the rules relating to indirect calls via function pointers.
These use the \emph{may call} relationship from \cref{ss:stc} to overapproximate the possible callees, and then apply the rules given above.
The constraints marked with $(*)$ and $(\ddagger)$ together with the definition of the summaries $(\dagger)$ describe the interprocedural aspects of the analysis:
Constraints of the form $(\ddagger)$ ensure that the information that an actual supplied at one of
the callsites may point to escaped memory is propagated to the callee function's corresponding
formal parameter. Summaries are computed by checking whether the formal parameters may potentially
point to escaped memory within the callee. These summaries are then used at callsites $(*)$ to ensure
that the actuals supplied are also marked as escaped. 
\begin{proposition}
    A memory access instruction $x \equiv (a_0 = r(*a)) \in \Instructions_f$  
    or $x \equiv w(*a, a_0) \in \Instructions_f$ need not be instrumented, i.e. $x\not\in\memInstrSuff$  if $a \not\in \DangerousRegisters_f$.
\end{proposition}
\begin{proofsketch}
    By verifying that if $a \not\in \DangerousRegisters_f$ implies that none of the memory locations
    pointed to by $a$ can fulfill the escape conditions given in \cref{ss:ea}. Then, by \cref{prop:ea1}, these memory locations remain thread-local, and thus accesses to them cannot be involved
    in any race by \cref{prop:ea0}.
\end{proofsketch}
\begin{remark}
    Our implementation leverages the \textsc{MemorySSA} analysis~\cite{novillo2007memory} as available
    in \textsc{LLVM} to efficiently track the flow of pointers through memory operations, meaning some of the rules above are implemented using \textsc{MemorySSA} queries rather than explicit iteration over all instructions.
    Furthermore, our implementation also handles dynamic memory allocation, which
    we have omitted from our core language.
\end{remark}

\section{Proofs for Dominance-Based Elimination}\label{app:de-proofs}
\begin{definition}[Redundancy]
    \label{def:redundancy}
    An access only involved in dominated races, is called
    \emph{redundant}. %
    The set of all redundant accesses in a program is called $\memInstrRed$.
\end{definition}
\begin{proposition}\label{prop:red-complete}
    A dynamic data race analyzer that is complete w.r.t.\ instrumented accesses, will, when
    instrumenting some superset of $\memInstrSuff \setminus \memInstrRed$, report all observed races that are not dominated.
\end{proposition}
\begin{proof}
    Directly from \cref{def:completeness,def:sufficient-instr-set,def:redundancy}.
\end{proof}

\begin{definition}[Release-like Instruction]
    An instruction is \emph{release-like} 
    if a corresponding dynamic event can appear on the left side of one of the orders
    whose transitive closure defines the happens-before relation (excluding $\to_{po}$).
    For our core language, this is \code{unlock} and \code{create} instructions.
\end{definition}

We mark $a$ as redundant if there is a memory access instruction $a_{\text{before}}$ such that the following conditions hold:
\begin{proposition}\label{prop:really-red}
    An access instruction $a$  for which there is another access instruction $a_{\text{before}}$ such that 
    \begin{enumerate}
    \item $\bblock\,a_{\text{before}} \text{ dominates } \bblock\,a$; and\label{item:dom-app}
    \item $\locs{a} = \locs{a_{\text{before}}} \land |\locs{a}| = 1$, i.e., both access instructions access the same (concrete!) memory location\footnote{
        In our implementation, we use the weaker (but still sufficient) requirement that
        the accessed locations \emph{must} alias, as computed by the \textsc{LLVM} framework.
    };\label{item:same-loc-app}
    \item if $a$ is a write, $a_{\text{before}}$ is also a write;\label{item:write-app}
    \item all paths from $a_\text{before}$ to $a$ through the \textbf{inter-procedural} control-flow-graph are free
    of release-like synchronization instructions and external calls.\label{item:no-sync-app}
\end{enumerate}
    all hold is redundant, i.e., $a \in \memInstrRed$.
\end{proposition}
\begin{proof}
    An access instruction $a$ is redundant if it is only involved in dominated races.
    Assume for a contradiction that there is an $a_{\text{before}}$ fulfilling conditions (1)-(4),
    but $a$ is involved in a data race $r$ involving another access, say $b$ where $r$ is not dominated.
    We show that under these circumstances, there also is a race between $a_{\text{before}}$ and $b$ which dominates
    $r$, contradicting the assumption that $r$ is not dominated.
    Consider some trace $t$ containing a race between dynamic events $e_a$ and $e_b$ corresponding to access
    instructions $a$ and $b$.
    This trace also contains a dynamic event $e_{a_{\text{before}}}$ corresponding to access $a_{\text{before}}$, as
    $\bblock\,a_{\text{before}}$ dominates $\bblock\,a$ (by (\ref{item:dom-app})).
    Also, by (\ref{item:same-loc-app}), both $e_a$ and $e_{a_{\text{before}}}$ access the same memory location, which, as
    $e_a$ and $e_b$ race, is the same as the location accessed by $e_b$.
    By (\ref{item:write-app}), if $e_a$ is a write, so is $e_{a_{\text{before}}}$. Thus, at least one of the accesses
    $e_{a_{\text{before}}}$ and $e_b$ is a write.
    By (\ref{item:no-sync-app}), all paths from $a_{\text{before}}$ to $a$ are free from release-like synchronization instructions,
    meaning there are no outgoing non-po happens before edges 
    for any of the events belonging to the same thread ordered between $e_{a_{\text{before}}}$ and $e_a$ by the program order.
    Thus, as $e_a$ and $e_b$ are unordered, so are $e_{a_{\text{before}}}$ and $e_b$.
    Thus, $e_{a_{\text{before}}}$ and $e_b$ constitute a race which dominates the race between $e_a$ and $e_b$,
    contradicting the assumption that $r$ is not dominated.
\end{proof}
\medskip

\begin{proposition}
    A dynamic data race analyzer that is complete w.r.t.\ instrumented accesses, will, when setting the instrumentation set to
    $\memInstr \supseteq \mathcal{W}$, and not aborting the execution of the program under
    test prematurely, report every observed race --- or an equivalent post- or pre-dominating race.
\end{proposition}
\begin{proof}
    Directly from \cref{def:completeness,def:weak-sufficient-instr-set}.
\end{proof}

\begin{proposition}\label{prop:post-red-complete}
    If $\mathcal{W}$ is weakly sufficient instrumentation for a program $P$,
    so is the set $\mathcal{W} \cap \memInstrSuff$.
\end{proposition}
\begin{proof}
    Consider some access instruction $i \in \mathcal{W}\setminus\memInstrSuff$.
    By \cref{def:sufficient-instr-set}, $i$ is not involved in any data races.
    Thus, removing $i$ from $\mathcal{W}$ preserves the weak sufficiency property.
\end{proof}

\medskip
\paragraph{Remark on \cref{def:post-domination}}
As explained in the main text, intuitively, this definition requires that, as soon as $r_1$ occurs,
$r_2$ must eventually also occur in the same execution. Technically, this is achieved by demanding
that every possible trace originating at $t$ (including $t$ itself) can be extended to contain $r_2$.
This excludes cases such as the program taking a different branch along which $r_2$ does not happen or 
program execution diverging after $r_1$, but before $r_2$ has happened.
In particular, executions that do not reach $r_2$ due to non-termination or premature termination (e.g., via
an explicit non-returning exit) are excluded by this requirement.

\begin{definition}[Acquire-like Instruction]
    An instruction is called \emph{acquire-like} 
    if a corresponding dynamic event can appear on the right side of one of the non-po orders
    whose transitive closure defines the happens-before relation.
    For our core language, this corresponds to \code{lock} instructions.\footnote{
        If the core language also contained thread-join operations, these would also be acquire-like instructions.
    }
\end{definition}

\begin{proposition}
    An access instruction $a$ for which there is another access instruction $a_{\text{after}}$ such that 
    \begin{enumerate}
    \item $\bblock\,a_{\text{after}} \text{ post-dominates } \bblock\,a$; and\label{item:dom-post-app}
    \item $\locs{a} = \locs{a_{\text{after}}} \land |\locs{a}| = 1$, i.e., both access instructions access the same (concrete!) memory location;\label{item:same-loc-post-app}
    \item if $a$ is a write, $a_{\text{after}}$ is also a write;\label{item:write-post-app}
    \item all paths from $a$ to $a_{\text{after}}$ through the control-flow-graph are free
    of acquire-like synchronization instructions and external calls;\label{item:no-sync-post-app}
    \item all paths from $a$ to $a_{\text{after}}$ do not contain any loops or function calls.\label{item:no-diverge-app}
\end{enumerate}
    all hold is only involved in post-dominated races.
\end{proposition}
\begin{proof}
    Assume for a contradiction that there is an $a_{\text{after}}$ fulfilling conditions 
    (1)-(5),
    but $a$ is involved in a non-post-dominated data race $r$ involving another access, say $b$.
    We show that under these circumstances, there also is a race between $a_{\text{after}}$ and $b$ which post-dominates
    $r$, contradicting the assumption that $r$ is non-post-dominated.
    Consider some trace $t$ containing a race between dynamic events $e_a$ and $e_b$ corresponding to access
    instructions $a$ and $b$.
    We first establish that every trace $t'$ that has $t$ as a prefix can be extended to a trace $t''$ containing
    a dynamic event $e_{a_{\text{after}}}$ corresponding to access instruction $a_{\text{after}}$ after $e_a$.
    If $t'$ already contains such an $e_{a_{\text{after}}}$, we are done.
    Consider now the case where this is not true. Such a trace then has a last event that belongs to the same 
    thread as $e_a$ that is either $e_a$ itself or a successor to $e_a$ in program order.
    As $\bblock\,a_{\text{after}}$ post-dominates $\bblock\,a$, every path from $\bblock\,a$ to the exit node
    of the CFG must contain $\bblock\,a_{\text{after}}$.
    As there can be no loops or function calls on any path from $a$ to $a_{\text{after}}$ (by (\ref{item:no-diverge-app})),
    and no (potentially) blocking operations can occur on these paths (by (\ref{item:no-sync-post-app})), $t'$ can always
    be extended to obtain a trace $t''$ that contains $e_{a_{\text{after}}}$ after $e_a$.
    It remains to show that $e_{a_{\text{after}}}$ and $e_b$ race. For the reasoning on the type of access and the memory 
    location see the proof of \cref{prop:really-red}.
    Thus, it remains to show that $e_{a_{\text{after}}}$ and $e_b$ are unordered.
    By (\ref{item:no-sync-post-app}), all paths from $a$ to $a_{\text{after}}$ are free from acquire-like synchronization instructions,
    meaning there are no incoming non-po happens before edges 
    for any of the events belonging to the same thread ordered between $e_{a}$ and $e_{a_{\text{after}}}$ by the program order.
    Thus, as $e_a$ and $e_b$ are unordered, so are $e_{a_{\text{after}}}$ and $e_b$.
    Thus, any trace with prefix $t$ can be extended to contain a race between $e_{a_{\text{after}}}$ and $e_b$,
    contradicting the assumption that $r$ is non-post-dominated. 
\end{proof}

\subsection{Better Detection of Single-Threaded Contexts}\label{app:stc-proof}
We recall \cref{th:stc-dyn}:
\begin{theorem}
    The transformation preserves the soundness and completeness properties of the dynamic race detector, provided
    it disregards any synchronization induced by accesses to the counter.
\end{theorem}
\begin{proof}[Proof sketch]
    The proof proceeds in two steps. First, we show that \code{__thread_cnt} has value
    $1$ only if there is only one thread active in the program.
    It then remains to show that the race detection algorithm does not depend on instrumentation calls for
    accesses that--- in the current run ---happen in a single-threaded context.
    This is easily justified as the access event
    cannot be involved in any race, and as ordinary accesses do not themselves introduce $\to_{hb}$ edges
    between different threads --- meaning the access need not be recorded, and no instrumentation call is necessary.
\end{proof}

\section{Impact of Thread Count on Optimization Effectiveness}\label{app:contention}
As shown in \cref{fig:results_runtime_memcached_contention,fig:results_runtime_sqlite_contention,fig:results_runtime_ffmpeg_contention}, our optimizations remain effective as the thread or CPU count rises. In Memcached, the speedup of \texttt{TSan+AllOpt} falls from $2.28\times$ at 2 server threads to $1.90\times$ at 8; it follows the server's CPU time per request, which \texttt{TSan+AllOpt} cuts by $2.30\times$ at 2 threads and by $2.02\times$ at 8. In FFmpeg, the speedup does not change from 2 to 16 encoder threads: most of the gain comes from the passthrough transcode, which does not use the encoder threads. In SQLite, run with 15 threads on 2 to 16 CPUs, \texttt{TSan+AllOpt} is faster at every CPU count, least so on 2 CPUs.

\pgfplotsset{
    myplotstyle/.style={
        ytick distance=0.1,
        yticklabel style={/pgf/number format/fixed, /pgf/number format/fixed zerofill, /pgf/number format/precision=1},
        xtick={2,6,10,14,18,22,26,30,34,38},
        line width=1pt,
        width=\linewidth,
        height=5cm,
        xlabel={Threads},
        ylabel={Speedup},
        ylabel near ticks,
        ylabel shift = -5pt,
        xlabel near ticks,
        label style={font=\small},
        tick label style={font=\footnotesize},
        title style={font=\small\bfseries}
    }
}

\pgfplotsset{
    ffmpegcontentionplotstyle/.style={
        title style={text height=1.6ex, text depth=0.3ex},
        yticklabel style={/pgf/number format/fixed, /pgf/number format/fixed zerofill, /pgf/number format/precision=1},
        xtick={2,6,10,14,18,22,26,30,34,38},
        line width=1pt,
        scale only axis,
        width=0.27\textwidth,
        height=3.4cm,
        xlabel={Threads},
        ylabel={Speedup},
        ylabel near ticks,
        ylabel shift = -5pt,
        xlabel near ticks,
        label style={font=\small},
        tick label style={font=\footnotesize},
        title style={font=\small\bfseries}
    }
}

\begin{figure*}
\centering
\begin{subfigure}[b]{0.495\textwidth}
\centering
\begin{tikzpicture}
  \begin{axis}[
    xmin=1.7, xmax=8.3,
    myplotstyle,
    xtick={2,4,6,8},
    legend columns=4,
    legend style={
      at={(0.5,-0.28)},
      anchor=north,
      font=\scriptsize,
      /tikz/every even column/.append style={column sep=0.08cm},
      /tikz/every odd column/.append style={column sep=0.08cm}
    },
    legend to name=contentionlegend
  ]
\addplot+[mark=pentagon*, color=cyan, densely dotted] coordinates {(2,2.225) (4,2.052) (6,1.937) (8,1.844)};\addlegendentry{TSan+EA};

\addlegendimage{mark=*, color=green!60!black, dash dot}\addlegendentry{TSan+LO};
\addlegendimage{mark=x, color=red!80!black, densely dashdotted}\addlegendentry{TSan+DynSTC};

\addplot+[mark=o, color=black, solid] coordinates {(2,2.277) (4,2.123) (6,1.993) (8,1.902)};

\addlegendentry{TSan+AllOpt};

\end{axis}
\end{tikzpicture}
\caption{Memcached}
\label{fig:results_runtime_memcached_contention}
\end{subfigure}
\hfill
\begin{subfigure}[b]{0.495\textwidth}
\centering
\begin{tikzpicture}
  \begin{axis}[
    xmin=1.5, xmax=16.5,
    myplotstyle,
    xtick={2,4,8,16},
    xlabel={CPUs},
  ]
\addplot+[mark=*, color=green!60!black, mark options={solid, fill=green!60!black}, dash dot] coordinates {(2,1.121) (4,1.342) (8,1.437) (16,1.342)};
\addplot+[mark=o, color=black, solid] coordinates {(2,1.093) (4,1.329) (8,1.472) (16,1.351)};
  \end{axis}
\end{tikzpicture}
\caption{SQLite}
\label{fig:results_runtime_sqlite_contention}
\end{subfigure}

\medskip

\begin{subfigure}[b]{\textwidth}
\centering
\begin{tikzpicture}
    \begin{groupplot}[
        group style={
            group size=3 by 1,
            horizontal sep=1.0cm,
        },
    ]

    \nextgroupplot[
        title={MJPEG},
        xmin=1.5, xmax=16.5,
        ymin=0.95, ymax=1.10,
        ytick distance=0.1,
        ffmpegcontentionplotstyle,
        xtick={2,4,8,16},
    ]
\addplot+[mark=x, color=red!80!black, densely dashdotted] coordinates {(2,1.037) (4,1.039) (8,1.050) (16,1.001)};
\addplot+[mark=o, color=black, solid] coordinates {(2,1.032) (4,1.031) (8,1.027) (16,1.001)};

    \nextgroupplot[
        title={H.264},
        xmin=1.5, xmax=16.5,
        ymin=0.95, ymax=1.10,
        ytick distance=0.1,
        ffmpegcontentionplotstyle,
        xtick={2,4,8,16},
        ylabel={},
    ]
\addplot+[mark=x, color=red!80!black, densely dashdotted] coordinates {(2,1.004) (4,1.012) (8,1.020) (16,1.000)};
\addplot+[mark=o, color=black, solid] coordinates {(2,0.997) (4,0.999) (8,1.003) (16,1.010)};

    \nextgroupplot[
        title={Passthrough},
        xmin=1.5, xmax=16.5,
        ymin=2.4, ymax=2.7,
        ytick distance=0.1,
        ffmpegcontentionplotstyle,
        xtick={2,4,8,16},
        ylabel={},
    ]
\addplot+[mark=x, color=red!80!black, densely dashdotted] coordinates {(2,2.534) (4,2.523) (8,2.518) (16,2.523)};
\addplot+[mark=o, color=black, solid] coordinates {(2,2.534) (4,2.523) (8,2.545) (16,2.523)};

    \end{groupplot}
\end{tikzpicture}
\caption{FFmpeg}
\label{fig:results_runtime_ffmpeg_contention}
\end{subfigure}

\pgfplotslegendfromname{contentionlegend}

\caption{Runtime Performance under Contention: Speedup (SU) vs. Original TSan.}
\Description{Line charts of the speedup over Original TSan against the thread or CPU count: Memcached with TSan+AllOpt and TSan+EA at 2 to 8 threads, SQLite with TSan+AllOpt and TSan+LO at 2 to 16 CPUs, and FFmpeg's MJPEG, H.264 and passthrough transcodes with TSan+AllOpt and TSan+DynSTC at 2 to 16 threads.}
\label{fig:results_runtime_contention}
\end{figure*}

\section{Additional Chromium Results}\label{app:chromium}

\Cref{fig:chromium_selected_appendix} shows speedups for three additional Chromium
benchmarks (Speedometer~3.0, Blink Image Decoder, and Blink SVG), complementing the Parser
and Paint plots in the main text.
\Cref{fig:chromium-cdf} shows the cumulative distribution of \code{TSan+AllOpt}
speedups across the 48 Blink benchmarks.

\begin{figure*}[t]
    \centering

    \pgfplotsset{
        appendixchromestyle/.style={
            yticklabel style={/pgf/number format/fixed, /pgf/number format/fixed zerofill, /pgf/number format/precision=1},
            ybar,
            width=5.35cm,
            height=3.8cm,
            bar width=0.30cm,
            bar shift=0pt,
            xtick={1,2,3,4,5,6,7,8},
            xticklabels={EA, LO, STC, SWMR, DE, DE+Peeling, FE-INL, AllOpt},
            xticklabel style={
                font=\tiny,
                anchor=east,
                rotate=45,
                align=right,
                inner sep=1pt
            },
            tick style={draw=none},
            ymin=0.95,
            ymax=1.60,
            ytick distance=0.1,
            ylabel style={font=\scriptsize},
            label style={font=\tiny},
            tick label style={font=\tiny},
            nodes near coords,
            nodes near coords style={
                font=\tiny,
                /pgf/number format/fixed,
                /pgf/number format/fixed zerofill,
                /pgf/number format/precision=2,
                anchor=south
            },
            clip=false,
            enlarge x limits={abs=0.5}
        }
    }

    \subfloat[Speedometer 3.0 Score\label{fig:chromium_selected_speedometer}]{
        \begin{minipage}{0.31\textwidth}\centering
        \begin{tikzpicture}
            \begin{axis}[appendixchromestyle, ymin=0.85, ymax=1.27, ytick distance=0.1, xtick={1,2,3,4,5,6,7,8}, xticklabels={EA, LO, STC, SWMR, DE, DE+Peeling, FE-INL, AllOpt}]
            \addplot[style={fill=blue!70, draw=none}, nodes near coords={}, error bars/.cd, y dir=both, y explicit]
                coordinates {(1,0.905) += (0,0.030) -= (0,0.035) (2,0.985) += (0,0.021) -= (0,0.020) (3,0.985) += (0,0.021) -= (0,0.020) (4,0.985) += (0,0.021) -= (0,0.020) (5,0.988) += (0,0.044) -= (0,0.054) (6,0.940) += (0,0.020) -= (0,0.023) (7,1.056) += (0,0.041) -= (0,0.053)};
            \addplot[style={fill=red!56, draw=none}, nodes near coords={}, error bars/.cd, y dir=both, y explicit]
                coordinates {(8,1.143) += (0,0.051) -= (0,0.052)};
            \addplot[draw=none, fill=none, point meta=explicit symbolic]
                coordinates {(1,0.934) [0.90] (2,1.006) [0.98] (3,1.006) [0.98] (4,1.006) [0.98] (5,1.032) [0.99] (6,0.959) [0.94] (7,1.097) [1.06] (8,1.194) [1.14]};
            \end{axis}
        \end{tikzpicture}
        \end{minipage}
    }
    \subfloat[Image Decoder: WebP\label{fig:chromium_selected_imgdec}]{
        \begin{minipage}{0.31\textwidth}\centering
        \begin{tikzpicture}
            \begin{axis}[appendixchromestyle, ymax=1.10, ytick distance=0.1, xtick={1,2,3,4,5,6,7,8}, xticklabels={EA, LO, STC, SWMR, DE, DE+Peeling, FE-INL, AllOpt}]
            \addplot[style={fill=blue!70, draw=none}, nodes near coords={}, error bars/.cd, y dir=both, y explicit]
                coordinates {(1,1.026) += (0,0.025) -= (0,0.022) (2,1.012) += (0,0.015) -= (0,0.012) (3,1.012) += (0,0.015) -= (0,0.012) (4,1.012) += (0,0.015) -= (0,0.012) (5,1.024) += (0,0.028) -= (0,0.023) (6,1.018) += (0,0.013) -= (0,0.011) (7,1.033) += (0,0.017) -= (0,0.015)};
            \addplot[style={fill=red!56, draw=none}, nodes near coords={}, error bars/.cd, y dir=both, y explicit]
                coordinates {(8,1.048) += (0,0.026) -= (0,0.026)};
            \addplot[draw=none, fill=none, point meta=explicit symbolic]
                coordinates {(1,1.051) [1.03] (2,1.027) [1.01] (3,1.027) [1.01] (4,1.027) [1.01] (5,1.053) [1.02] (6,1.031) [1.02] (7,1.050) [1.03] (8,1.074) [1.05]};
            \end{axis}
        \end{tikzpicture}
        \end{minipage}
    }
    \subfloat[SVG: SvgCubics\label{fig:chromium_selected_svg}]{
        \begin{minipage}{0.31\textwidth}\centering
        \begin{tikzpicture}
            \begin{axis}[appendixchromestyle, ymin=0.9, ymax=1.5, xtick={1,2,3,4,5,6,7,8}, xticklabels={EA, LO, STC, SWMR, DE, DE+Peeling, FE-INL, AllOpt}]
            \addplot[style={fill=blue!70, draw=none}, nodes near coords={}, error bars/.cd, y dir=both, y explicit]
                coordinates {(1,0.998) += (0,0.014) -= (0,0.015) (2,0.984) += (0,0.020) -= (0,0.022) (3,0.984) += (0,0.020) -= (0,0.022) (4,0.984) += (0,0.020) -= (0,0.022) (5,0.974) += (0,0.019) -= (0,0.023) (6,0.991) += (0,0.017) -= (0,0.017) (7,1.084) += (0,0.019) -= (0,0.026)};
            \addplot[style={fill=red!56, draw=none}, nodes near coords={}, error bars/.cd, y dir=both, y explicit]
                coordinates {(8,1.366) += (0,0.030) -= (0,0.030)};
            \addplot[draw=none, fill=none, point meta=explicit symbolic]
                coordinates {(1,1.012) [1.00] (2,1.004) [0.98] (3,1.004) [0.98] (4,1.004) [0.98] (5,0.994) [0.97] (6,1.008) [0.99] (7,1.103) [1.08] (8,1.396) [1.37]};
            \end{axis}
        \end{tikzpicture}
        \end{minipage}
    }

    \caption{Additional representative Chromium benchmarks (Speedup vs. Original TSan). Higher is better. These examples complement the Parser and Paint plots shown in the main text.}
    \Description{Three additional Chromium bar charts in the appendix for Speedometer 3.0 Score, Blink Image Decoder WebP, and Blink SVG SvgCubics. Each chart compares the single configurations with AllOpt.}
    \label{fig:chromium_selected_appendix}
\end{figure*}

\begin{figure*}
    \centering
    \begin{tikzpicture}
        \begin{axis}[
            width=\linewidth,
            height=5cm,
            xlabel={\small Speedup ($\le X$)},
            ylabel={\small Fraction of Benchmarks},
            ylabel near ticks,
            ymin=0, ymax=1.05,
            xmin=0.9, xmax=1.9,
            xticklabel style={/pgf/number format/fixed, /pgf/number format/fixed zerofill, /pgf/number format/precision=1}, yticklabel style={/pgf/number format/fixed, /pgf/number format/fixed zerofill, /pgf/number format/precision=1},
            tick label style={font=\footnotesize},
            label style={font=\scriptsize},
            legend style={font=\scriptsize},
            legend pos=south east,
            legend cell align={left},
        ]
        \addplot[thick, blue!80!black, mark=none, const plot] coordinates {
        (0.98, 0.00) (0.983, 0.0208) (1.000, 0.0417) (1.015, 0.0625)
        (1.018, 0.0833) (1.026, 0.1042) (1.042, 0.1250) (1.044, 0.1458)
        (1.051, 0.1667) (1.084, 0.1875) (1.105, 0.2083) (1.116, 0.2292)
        (1.130, 0.2500) (1.153, 0.2708) (1.160, 0.2917) (1.161, 0.3125)
        (1.164, 0.3333) (1.173, 0.3542) (1.177, 0.3750) (1.185, 0.3958)
        (1.185, 0.4167) (1.224, 0.4375) (1.231, 0.4583) (1.236, 0.4792)
        (1.259, 0.5000) (1.260, 0.5208) (1.268, 0.5417) (1.271, 0.5625)
        (1.273, 0.5833) (1.276, 0.6042) (1.282, 0.6250) (1.288, 0.6458)
        (1.292, 0.6667) (1.293, 0.6875) (1.302, 0.7083) (1.302, 0.7292)
        (1.313, 0.7500) (1.330, 0.7708) (1.366, 0.7917) (1.373, 0.8125)
        (1.387, 0.8333) (1.408, 0.8542) (1.413, 0.8750) (1.553, 0.8958)
        (1.561, 0.9167) (1.576, 0.9375) (1.664, 0.9583) (1.693, 0.9792)
        (1.867, 1.0000)
        };
        \addlegendentry{48 Blink benchmarks}

        \draw[black, dashed, line width=0.5pt] (axis cs:1.0,0) -- (axis cs:1.0,1);
        \node[anchor=west, black, font=\scriptsize, rotate=90, yshift=-5pt, xshift=45pt] at (axis cs:1.0, 0.05) {Baseline};

        \draw[black, dotted, line width=0.5pt] (axis cs:0,0.5) -- (axis cs:1.9,0.5);
        \node[anchor=south east, black, font=\scriptsize] at (axis cs:1.9, 0.52) {Median};

        \end{axis}
    \end{tikzpicture}
    \caption{Cumulative distribution of \code{TSan+AllOpt} speedups across the 48 Blink benchmarks.
    The median speedup is $1.26\times$; ${\sim}98\%$ of benchmarks show speedup $\ge 1.0\times$.}
    \Description{Step plot of the fraction of the 48 Blink benchmarks whose TSan+AllOpt speedup over Original TSan is at most a given value, with a dashed vertical line at a speedup of 1 and a dotted horizontal line at 0.5 marking the median. Nearly the whole curve lies to the right of 1.}
    \label{fig:chromium-cdf}
\end{figure*}

\section{Details on the Oracle Tracer}\label{app:oracle-tracer}
We want to estimate how many \emph{dynamic} accesses need a check.
We develop a dynamic tracer that identifies all dynamic events that access shared memory and happen while the program is multi-threaded.

The oracle reports that the fraction of dynamic accesses that both happen in a multi-threaded phase and touch memory accessed by multiple threads is 23.28\% for Memcached, 95.64\% for Redis, 26.31\% for FFmpeg, 64.99\% for SQLite, and 30.06\% for MySQL.
For Memcached, FFmpeg, and MySQL, the oracle therefore suggests that only a minority of dynamic accesses are potentially relevant, whereas SQLite exhibits a substantially larger shared-access footprint.
While the analyses need to consider \emph{all possible executions} to ensure soundness, these results still highlight room
for improvement in analysis precision in some applications, and highlight the challenges of relating static instrumentation decisions to dynamic access frequencies in concurrent programs.
We identified five primary causes of static analysis imprecision, summarized in \cref{tab:imprecision_causes}.
On top of providing insights into the limitations of our static analyses, they
may serve as a roadmap for future enhancements.

\begin{table*}
  \centering
  \caption{Primary Causes of Static Analysis Imprecision.}
  \label{tab:imprecision_causes}
  \small
  \setlength{\tabcolsep}{4pt}
  \begin{tabular}{@{}p{0.21\linewidth} p{0.58\linewidth} p{0.13\linewidth}@{}}
    \toprule
    \textbf{Cause} & \textbf{Description} & \textbf{Example} \\
    \midrule
    \textbf{Complex Data Flow} & Objects passed through complex function chain, making it difficult
    to prove that they do not escape. &
    Memcached,\newline Redis \\
    \addlinespace
    \textbf{Globally Reachable Heap} & Heap objects stored in global structures but used by one thread are treated as escaping. & SQLite,\newline Memcached \\
    \addlinespace
    \textbf{Function Pointers} & Indirect calls obscure control flow, forcing worst-case sharing assumptions. & Memcached \\
    \addlinespace
    \textbf{Runtime-Dependent\newline Addressing} & Runtime-computed indices or pointers prevent compile-time resolution. & SQLite \\
    \addlinespace
    \textbf{Inline Assembly} & \code{__asm__} blocks may access any memory and perform synchronization. & FFmpeg \\
    \bottomrule
  \end{tabular}
\end{table*}

\section{Details on the Performance of the \textsc{ReX} Dynamic Filter}\label{app:rex-details}
\cref{tab:rex_effectiveness} summarizes the effectiveness of our \textsc{ReX} implementation in filtering redundant memory accesses across the three benchmarks. The data confirms that \textsc{ReX} is highly effective at identifying redundancy, with over 95\% of accesses filtered in Memcached and SQLite, and around 37\% in Redis. Additionally, an overwhelming majority of the filtered accesses are intra-thread, indicating that they arise from repeated execution of code within loops.

\begin{table*}
  \centering
  \caption{Effectiveness of the \textsc{ReX} dynamic filter per benchmark.}
  \label{tab:rex_effectiveness}
  \begin{tabular}{@{}lc
    >{\centering\arraybackslash}p{2cm}
    >{\centering\arraybackslash}p{5cm}
    @{}}
    \toprule
    \textbf{Benchmark} & \textbf{Total Accesses} & \textbf{Filtered} & \textbf{Intra-thread (\% of Filtered)}  \\
    \midrule
    Memcached & 385M & 99\% & 86\% \\
    SQLite    & 168M & 96\% & 95\% \\
    Redis     & 12M & 37\% & 99\% \\
    \bottomrule
  \end{tabular}
\end{table*}

Despite these impressive results, this did not translate into a net performance gain in our implementation. In fact, the filter itself introduced a significant overhead, causing the benchmarks to run several times slower than even the original \tsan.
We attribute this to the inherent cost of the filtering logic. Although our implementation utilizes optimized synchronization with lock-free read paths to minimize contention, the maintenance of the shared concurrency history trie inevitably incurs substantial cache thrashing and atomic overhead on every memory access. This indicates that the performance bottleneck is fundamental to the dynamic filtering approach rather than an artifact of a specific implementation.
Optimizing this dynamic filtering mechanism to its full potential is a considerable engineering challenge. Since our primary goal was to use \textsc{ReX} as an analytical tool to quantify dynamic redundancy rather than to build a production-ready dynamic filter, we did not pursue these optimizations further.

\end{document}